\pdfoutput=1
\documentclass[letterpaper, 10 pt, conference]{ieeeconf}

\usepackage{todonotes}
\usepackage{bbm}

\usepackage{graphicx}
\usepackage{booktabs}
\usepackage{adjustbox}
\usepackage{wrapfig}
\usepackage{lipsum}
\usepackage{amsmath,amssymb,amsfonts}
\usepackage{siunitx}
\usepackage{mathtools}
\usepackage{xcolor}
\usepackage[colorlinks = true,
            linkcolor = blue,
            urlcolor  = blue,
            citecolor = blue,
            anchorcolor = blue]{hyperref}
\let\labelindent\relax
\usepackage[inline]{enumitem}
\usepackage[ruled,vlined,linesnumbered]{algorithm2e}
\usepackage[noend]{algpseudocode}
\usepackage{ifthen}
\usepackage{dsfont}
\usepackage[backend=biber,style=numeric-comp,sorting=none,maxbibnames=99]{biblatex}
\usepackage{balance}
\usepackage{dirtytalk}
\usepackage{afterpage}
\usepackage{times}
\usepackage{optidef}
\usepackage[capitalise]{cleveref}
\usepackage{enumitem}

\allowdisplaybreaks

\newcommand{\diag}{\text{diag}}
\newcommand{\E}{\mathbb{E}}

\newcommand{\M}{\mathcal{M}}
\newcommand{\N}{\mathbb{N}}

\newcommand{\Prob}{\mathbb{P}}

\newcommand{\ra}{\rightarrow}
\newcommand{\R}{\mathbb{R}}

\newtheorem{remark}{Remark}
\newtheorem{assumption}{Assumption}
\newtheorem{proposition}{Proposition}

\newtheorem{lemma}{Lemma}
\newtheorem{theorem}{Theorem}

\DeclarePairedDelimiterX\loro[1]\rbrack\lbrack{#1}
\DeclarePairedDelimiterX\lorc[1]\rbrack\rbrack{#1}
\DeclarePairedDelimiterX\lcro[1]\lbrack\lbrack{#1}
\DeclarePairedDelimiterX\lcrc[1]\lbrack\rbrack{#1}
\DeclarePairedDelimiterX\set[1]\lbrace\rbrace{#1}

\DeclarePairedDelimiterX\norm[1]\lVert\rVert{#1}

\newcommand{\mc}[1]{\mathcal{#1}}

\newcommand{\G}{\mathcal{G}}
\newcommand{\network}{\mathcal{N}}

\newcommand{\nodes}{V}
\newcommand{\edges}{E}
\newcommand{\paths}{\mathcal{P}}

\newcommand{\NE}{\text{NE}}

\newcommand{\boldlambda}{{\boldsymbol{\lambda}}}
\newcommand{\boldnu}{{\boldsymbol{\nu}}}
\newcommand{\boldalpha}{{\boldsymbol{\alpha}}}
\newcommand{\boldtau}{{\boldsymbol{\tau}}}
\newcommand{\boldbeta}{{\boldsymbol{\beta}}}
\newcommand{\boldeta}{{\boldsymbol{\eta}}}

\newcommand{\boldv}{\textbf{v}}
\newcommand{\boldf}{\boldsymbol{f}}
\newcommand{\tr}{\text{tr}}

\newcommand{\rangespace}{\mathcal{R}}
\newcommand{\nullspace}{\mathcal{N}}
\newcommand{\lagrangian}{\mathcal{L}}

\newcommand{\tollSet}{\mathcal{T}}
\newcommand{\tollSetInterior}{\tilde{\mathcal{T}}}

\newcommand{\routingMatrix}{R}
\newcommand{\flow}{f}
\newcommand{\flowPath}{\hat{\flow}}
\newcommand{\flowVec}{\boldf}
\newcommand{\flowVecPath}{\hat{\flowVec}}
\newcommand{\demand}{\eta}
\newcommand{\demandVec}{\boldeta}
\newcommand{\setFeasibleFlows}{F}
\newcommand{\setFeasibleFlowsPath}{\hat{\setFeasibleFlows}}
\newcommand{\setNashFlows}{\setFeasibleFlows^\NE}
\newcommand{\flowNash}{\flow^\NE}

\newcommand{\flowVecNash}{\flowVec^\NE}
\newcommand{\flowVecNashPath}{\flowVecPath^\NE}
\newcommand{\latency}{\ell}
\newcommand{\latencyVec}{\boldsymbol{\latency}}
\newcommand{\latencySlope}{\beta}
\newcommand{\latencySlopeVec}{\boldbeta}
\newcommand{\latencySlopeMatrix}{B}
\newcommand{\slantProjMatrix}{\Gamma}
\newcommand{\disturb}{\alpha}
\newcommand{\disturbVec}{\boldalpha}
\newcommand{\toll}{\tau}
\newcommand{\tollVec}{\boldtau}

\newcommand{\gelbrichBall}{\text{GB}}

\newcommand{\gameCongestion}{\G}
\newcommand{\rosenthalPotential}{\Phi}
\newcommand{\latencyWorstEq}{g}
\newcommand{\dualIneq}{\lambda}
\newcommand{\dualIneqVec}{\boldlambda}
\newcommand{\dualEq}{\nu}
\newcommand{\dualEqVec}{\boldnu}
\newcommand{\meanDisturb}{\boldsymbol{\theta}}
\newcommand{\covDisturb}{\Sigma}
\newcommand{\meanDisturbNominal}{\hat{\meanDisturb}}
\newcommand{\covDisturbNominal}{\hat \covDisturb}

\newcommand{\latencyWorstSlopeWrtDisturb}{q}
\newcommand{\latencyWorstOffsetWrtDisturb}{q_0}
\newcommand{\disturbDistr}{\mu}
\newcommand{\disturbDistrNominal}{\hat\disturbDistr}
\newcommand{\disturbRange}{\delta}
\newcommand{\disturbDistrDeviation}{\epsilon}
\newcommand{\disturbDistrDeviationMax}{\epsilon_{\max}}
\newcommand{\optTollSet}{S_\tollVec^\star}

\newcommand{\latencyData}{y}
\newcommand{\latencyDataVec}{\boldsymbol{\latencyData}}

\newcommand{\systemLatency}{L}

\newcommand{\onesVector}{\textbf{1}}

\newcommand{\paren}[1]{{({#1})}}

\newcommand{\eventDisburbanceBound}{\mathcal{E}}

\bibliography{refs}

\IEEEoverridecommandlockouts

\title{
Distributionally Robust Tolls for Traffic Networks with Affine Latency Functions
}

\usepackage{pgf}

\pgfmathsetseed{\number\pdfrandomseed}
\pgfmathrandominteger{\rand}{0}{1}

\iftrue
    \author{
Chih-Yuan Chiu$^1$, 
Sarah H. Q. Li$^{2\star}$,
and Bryce L. Ferguson$^{3\star}$
\thanks{${}^\star$Equal senior authorship.}
\thanks{$^{1}$Chih-Yuan Chiu is with the School of 
ECE, Georgia Tech,
Atlanta, GA, USA (\texttt{cyc at gatech dot edu}).}
\thanks{$^{2}$Sarah H.Q. Li is with the School of 
AE, Georgia Tech,
GA, USA (\texttt{sarahli at gatech dot edu})}
\thanks{$^{3}$Bryce L. Ferguson is with the Thayer School of Engineering at Dartmouth College, NH 03755, USA (\texttt{Bryce.L.Ferguson at dartmouth dot edu}).}
}
\else
    \author{
Chih-Yuan Chiu$^1$, 
Bryce L. Ferguson$^2$,
and Sarah H. Q. Li$^3$
\thanks{$^{1}$Chih-Yuan Chiu is with the School of Electrical and Computer Engineering, Georgia Institute of Technology, Atlanta, GA, USA (\texttt{cyc at gatech dot edu}).}
\thanks{$^{2}$Bryce L. Ferguson is with the Thayer School of Engineering at Dartmouth College, NH 03755, USA (\texttt{Bryce.L.Ferguson at dartmouth dot edu}).}
\thanks{$^{3}$Sarah H.Q. Li is with the School of Aerospace Engineering, Georgia Institute of Technology, Atlanta, GA, USA (\texttt{sarahli at gatech dot edu})}
}
\fi

\def\BibTeX{{\rm B\kern-.05em{\sc i\kern-.025em b}\kern-.08em
T\kern-.1667em\lower.7ex\hbox{E}\kern-.125emX}}

\def\BibTeX{{\rm B\kern-.05em{\sc i\kern-.025em b}\kern-.08em
T\kern-.1667em\lower.7ex\hbox{E}\kern-.125emX}}

\begin{document}

\maketitle

\thispagestyle{empty}
\pagestyle{empty}




\begin{abstract}


In network congestion games, system operators often utilize latency models, estimated from real-world traffic flow and travel time data, to design monetary incentives which steer equilibrium user behaviors towards 
lowering system-wide latency.
This work studies the impact of 
latency model uncertainty
when designing incentives in 
non-atomic network congestion games.
Our 
approach leverages
distributionally robust optimization (DRO), 
which captures data-driven uncertainty in latency models by considering worst-case distribution shifts.
We prove that, under mild and practically relevant assumptions, the distributionally robust tolling problem in single origin-destination, affine-latency congestion games can be solved via convex programming.
Numerical simulations illustrate 
that tolls designed to be distributionally robust against unknown disturbances can outperform tolls designed using fixed, nominal disturbance models in minimizing system-wide latency.
\end{abstract}




\section{Introduction}
\label{sec: Introduction}

The efficient operation of cyber-physical systems is critical to a functional modern society, and is increasingly regulated via optimization-driven pricing and tolls, such as congestion-based  pricing or marginal-pricing over traffic \cite{Paccagnan2019IncentivizingEfficientUse, Ratliff2018, Roughgarden2010AlgorithmicGameTheory, Yang1998PrincipleOfMarginalCostPricing,li2023adaptive}, communication \cite{Niyato2008}, and power \cite{Tang2019} systems.
Such pricing schemes often rely upon congestion game-based models to capture the impact of strategic users' resource usage on network congestion levels, and vice versa.
Unfortunately, the interplay between users' resource consumption and congestion levels varies with fluctuating demand and exogenous disturbances, and can only be imperfectly estimated through limited, noise-corrupted data, such as historical traffic counts, travel time measurements, and GPS traces \cite{Kim2024EstimateThenPredict}.
Even when analytical structures are imposed to describe congestion behavior, their parameters reflect empirical measurements and are subject to estimation error and variability. 
As a result, toll design in networked congestion games is inherently subjected to data-driven and time-varying uncertainty.

Despite the mismatch between the predictions of empirically estimated models and real world system outcomes, most toll design frameworks still hinge upon the availability of a known, static congestion model \cite{Paccagnan2021OptimalTaxesinAtomicCongestionGames, Yang1998PrincipleOfMarginalCostPricing, Roughgarden2010AlgorithmicGameTheory}.
Unfortunately, the performance of toll designs that fail to account for system parameter variability may degrade under significant parameter fluctuations
\cite{ChiuFerguson2025RobustnessOfIncentiveMechanisms}.
To address this issue, this work presents a framework for designing \emph{distributionally robust tolls} over network congestion games that optimize the system-wide latency under bounded distributional uncertainty in latency attributes.
Our framework enables incentive designers to incorporate considerations of distributional uncertainty when designing pricing mechanisms in real-world transportation contexts.
For instance, by modeling disturbances over edge latencies as random variables with data-driven probability distributions 
that capture weather fluctuations, a toll designer can design congestion pricing policies that account for worst-case road conditions.
While prior work \cite{ChiuFerguson2025RobustnessOfIncentiveMechanisms} studied the robustness of toll designs under \textit{deterministic} latency model errors, we consider misspecifications of \textit{stochastic} latency attributes stemming from distributional shifts. 
Leveraging distributionally robust optimization (DRO) 
\cite{Kuhn2019WassersteinDistributionallyRobustOptimizationTheoryApplicationsML}, we derive tolls that minimize expected latency under the worst-case disturbance distribution near a nominal prior. 
To our knowledge, our work is the first to adopt DRO techniques for incentive design in congestion games, although similar methods have been applied to study auction pricing \cite{suzdaltsev_distributionally_2022}, computing stochastic traffic equilibria \cite{Ahipasaoglu2019DistributionallyRobustMTE,marecek_distributionally_2017}, and non-game-theoretic toll design \cite{dokka_robust_2017}.

Our contribution is threefold. First, we formulate a variant of the toll design problem over network 
congestion
games which aims to minimize the worst-case system latency under equilibrium flows while considering distributional uncertainty over network edge latencies learned from imprecise traffic data.
Concretely, we leverage Wasserstein-based and Gelbrich hull-based uncertainty sets to capture ambiguity in the empirical disturbance distribution without reference to any specific parametric family of distributions.
Second, under standard assumptions on the latency functions, we present an analytic expression of the worst-case equilibrium system latency as a function of tolls and uncertain congestion parameters, which enables us to recast our distributionally robust toll design problem as a convex program.
Third,
our formulation reveals how statistical variability in the congestion parameters and structural properties of the latency model jointly influence both representability and solvability of the design problem.
Together, these results provide a principled framework for computing tolls that incentivize Nash equilibria with improved system-level performance while remaining robust to distributional, data-driven uncertainty.

\paragraph{Notation} 
For any positive integers $m, n \in \N$, we define $[n] := \{1, \cdots, n\}$,
and 
$\R_{\geq 0}^n := \{x \in \R^n: x_i \geq 0 \ \forall i \in [n] \}$, and $O_{m \times n} :=$ the $m \times n$ zero matrix, with indices omitted when clear from context.
For any set $S$, let $|S|$ and $2^S$ respectively denote the cardinality and power set of $S$.
Define $\onesVector_n := (1, 1, \cdots, 1) \in \R^n$. 
Given $\boldv = (v_1, \cdots, v_n) \in \R^n$, we define $\diag\{\boldv\} \in \R^{n \times n}$ to be the diagonal matrix whose $i$-th entry is $v_i$ for each $i \in [n]$. 
Given a positive semidefinite matrix $P$, we denote the unique positive semidefinite square root of $P$ by $P^{1/2}$, i.e., $(P^{1/2})^2 = P$, and we write $P^{-1/2} = (P^{1/2})^{-1}$.
We define $\mathds{1}\{\cdot\}$ to be the indicator function, which returns 1 when the input argument is true, and 0 otherwise.
Let $\mc{P}(\Omega)$ denote the set of all continuous, non-negative, and Lebesgue integrable probability density functions on $\Omega$ and $\mc{X}_{\Omega}$ to denote the set of corresponding random variables.


\section{Problem Setup}
\label{sec: Preliminaries}


\subsection{Network Structure}
\label{subsec: Network Structure}


Consider a traffic network $\network = (\nodes, \edges)$ consisting of $n_\nodes$ nodes and $n_\edges$ edges, 
with
$\nodes = [n_\nodes]$ and $\edges = [n_\edges]$; here, $\edges$ describes a set of directed edges between nodes. For each node $i \in \nodes$, let $\edges_i^-$ and $\edges_i^+$ respectively denote the sets of incoming and outgoing edges at node $i$. For each edge $e \in \edges$, let $i_e \in \nodes$ and $j_e \in \nodes$ respectively denote the start and end nodes of edge $e$. 
In our work, we assume that the network $\network$ is \textit{single-origin single-destination,} i.e., there exists unique $s, d \in \nodes$ such that $\edges_s^- = \edges_d^+ = \emptyset$, where $\emptyset$ denotes the empty set. 
Without loss of generality, we assume that $s = 1$ and $d = n_V$, i.e., the source and destination nodes are respectively indexed first and last in the node set $\nodes$.
We also
assume that $\network$ is \textit{acyclic,} i.e., there exists no finite subset of edges $\{e_1, \cdots, e_n\} \subseteq \nodes$ such that $j_{e_k} = i_{e_{k+1}}$ for each $k \in [n-1]$ and $i_{e_1} = j_{e_n}$. Next, we call an ordered subset of edges $P := \{e_1, \cdots, e_n\} \subseteq E$ a path through the network $\network$ if $i_{e_1} = s$, $j_{e_n} = d$, and $j_{e_k} = i_{p_{k-1}}$ for each $k \in [n-1]$, and let $\paths \in 2^{|\edges|}$ denote the set of all paths.

We assume that the network $\network$ is traversed by a homogeneous user population with demand $\demand_s > 0$, thus generating a set $\setFeasibleFlowsPath$ of \textit{feasible path flows} 
given by:
\begin{align}
    \setFeasibleFlowsPath &:= \Bigg\{ \flowVecPath \in \R_{\geq 0}^{|\paths|}: \hspace{5mm} \sum_{p \in \paths} \flow_P = \demand_s \Bigg\}.
\end{align}
Each feasible path flow $\flowVecPath \in \setFeasibleFlowsPath$ then induces a corresponding edge flow $\flowVec \in \R^{|\edges|}$, defined coordinate-wise by $\flow_e := \sum_{p: e \in p} \flowPath_P$ for each $e \in \edges$, which is contained in the set $\setFeasibleFlows$ of \textit{feasible (edge) flows} given by:
\begin{align} \label{Eqn: Feasible Flows, def}
    \setFeasibleFlows &:= \Bigg\{ \flowVec \in \R_{\geq 0}^{|\edges|}: \hspace{3mm} \sum_{e \in \edges_i^+} \flow_e = \sum_{e \in \edges_i^-} \flow_e, \ \forall i \in \nodes \backslash \{s\}, \\ \nonumber
    &\hspace{2.7cm} \sum_{e \in \edges_s^+} \flow_e = \demand_s \Bigg\}.
\end{align}
Let $\routingMatrix \in \R^{|\nodes| - 1} \times |\edges|$ and $\demandVec \in \R^{|\nodes|-1}$ be defined component-wise as follows. For each $i \in \nodes \backslash \{d\}$ and $e \in \edges$:
\begin{align}
    \routingMatrix_{ie} &:= \mathds{1}\{e \in \edges_i^+\} - \mathds{1}\{e \in \edges_i^-\}, \\
    \demand_i &:= \demand_s \cdot \mathds{1}\{i = s\}. 
\end{align}
Then the feasible (edge) flow set \eqref{Eqn: Feasible Flows, def} can be written as:
\begin{align}
    \setFeasibleFlows = \{\flowVec \in \R_{\geq 0}^{|\edges|}: \routingMatrix \flowVec = \demandVec. \}
\end{align}

\subsection{Selfish Routing and Incentive Design}
We associate each edge $e \in E$ with a convex, strictly positive, and strictly increasing \textit{latency function} $\latency_e: \R_{\geq 0} \ra \R_{\geq 0}$, as well as a real-valued disturbance random variable $\disturb_e$ which captures random latency fluctuations.
As a mechanism to alter the perceived cost of each edge $e \in E$, a non-negative toll value $\toll_e \geq 0$ can be deployed by a system designer.
We define $\disturbVec := (\disturb_e: e \in E)$ (taking values in $\R^{|\edges|}$) and $\tollVec := (\toll_e: e \in E) \in \R_{\geq 0}^{|\edges|}$ for short.
We assume that users respond to 
the disturbance $\disturbVec$, whereas the toll designer only possesses a belief $\disturbDistr$ of the disturbance distribution.
For a given realization of $\disturbVec$, the resulting congestion game formulated using the network structure, latency, toll, and disturbance model specified above is denoted by $\gameCongestion(\disturbVec, \tollVec)$.

Next, we formulate the objectives, and analyze the decisions, of the user population and the traffic controller in a network congestion game with tolls $\gameCongestion(\disturbVec, \tollVec)$. 
During route selection,
we assume each user seeks to minimize the sum of their observed travel latency and monetary toll costs.
Nash equilibrium flow concepts capture the collective behavior emerging from these selfish decisions
\cite{Sandholm2010PopulationGamesAndEvolutionaryDynamics, Roughgarden2010AlgorithmicGameTheory}.
Concretely, we call a path flow $\flowVecNashPath \in \setFeasibleFlowsPath$ a \textit{Nash equilibrium path flow} of $\gameCongestion(\disturbVec, \tollVec)$,
with a corresponding (\textit{Nash equilibrium}) \textit{edge flow} $\flowVecNash \in \setFeasibleFlows$,
if each utilized path is of minimal cost under $\flowVecNashPath$, i.e., for each $P \in \paths$, we have $\flowPath_P > 0$ only if:
\begin{align}
\label{Eqn: Nash Equilibrium (path) flow, Def}
    P &\in \arg\min_{P' \in \paths} \sum_{e \in P'} \big[ \latency_e(\flowNash_e) + \toll_e \big].
\end{align}
We denote by $\setNashFlows(\disturbVec, \tollVec) \subseteq \R^{|\edges|}$ the set of all Nash equilibrium edge flows corresponding to the game $\gameCongestion(\disturbVec, \tollVec)$. If the game $\gameCongestion(\disturbVec, \tollVec)$ admits a unique Nash equilibrium edge flow, 
we write its Nash equilibrium edge flow by $\flowVecNash(\disturbVec, \tollVec)$.
From this point onward, \say{flow} will refer to edge flow (rather than path flow) unless otherwise stated.

It is well-known that the set of Nash equilibrium flows for a given congestion game with tolls $\gameCongestion(\disturbVec, \tollVec)$, can be characterized via convex programming \cite{Rosenthal1973ClassofGames}. Concretely, we define the \textit{Rosenthal potential function} 
$\rosenthalPotential(\cdot; \disturbVec, \tollVec): \R_{\geq 0}^{|\edges|} \ra \R$ of the game $\gameCongestion(\disturbVec, \tollVec)$ by:
\begin{align}
    \label{Eqn: Rosenthal Potential Function}
    \rosenthalPotential(\flowVec; \disturbVec, \tollVec) &:= \sum_{e \in E} \int_0^{\flow_e} \big[\latency_e(w) + \toll_e + \disturb_e \big] dw.
\end{align}
We then have:
\begin{align}
\label{Eqn: Nash Computation via Rosenthal Potential Function}
    \setNashFlows(\disturbVec, \tollVec) &= \arg \min_{\flowVec \in \setFeasibleFlows} \rosenthalPotential(\flowVec; \disturbVec, \tollVec).
\end{align}
%

Given flows $\flowVec$ and disturbances $\disturbVec$, we define the system-wide latency $\systemLatency: \R^{2|\edges|} \ra \R$, which describes aggregate user costs, as:
\begin{align}
    \systemLatency(\flowVec; \disturbVec) := \sum_{e \in E} \flow_e \big[\latency_e(\flow_e) + \disturb_e \big].
\end{align}
A system operator who observes the disturbance $\disturbVec$ perfectly would aim to designs tolls $\tollVec$ minimize the worst-case system-wide latency at equilibrium, i.e.,
\begin{align} \label{Eqn: (g) Worst-case System Latency at Equilibrium, def}
    \latencyWorstEq(\disturbVec, \tollVec) &:= \max_{\flowVec \in \setNashFlows(\disturbVec, \tollVec)} \systemLatency(\flowVec; \disturbVec).
\end{align}

However, as mentioned above, we assume that the system designer cannot adapt its designed toll to the precise 
realizations of
$\disturbVec$, but must design tolls based only on 
a belief distribution $\disturbDistr$ over $\disturbVec$.
Concretely, in the \textit{nominal toll design problem}, the system designer aims to design tolls $\tollVec$ to minimize the worst-case expected total latency at equilibrium:
\begin{align} \label{Eqn: Optimal nominal toll in preliminaries}
    \text{arg}\min_{\tollVec \in \R_{\geq 0}^{|\edges|}} \E_{\disturbVec \sim \disturbDistr}[\latencyWorstEq(\disturbVec, \tollVec)].
\end{align}
In this work, we are interested in the setting where the system designer's belief $\disturbDistr$ is inaccurate and may not match the true distribution of $\disturbVec$, possibly caused by approximation error or insufficient data for empirical estimation.
We thus seek to solve a new toll design problem where the resulting tolls must be robust to this form of distributional misspecification.

\subsection{Modeling Congestion in Latency via Empirical Data}
We present tractable models from the distributional robustness literature that enable toll designers to explicitly account for possible discrepancies between their empirical model of the disturbance distribution 
and the true distribution.
We make the following assumptions: 
\begin{enumerate}
    \item The toll designer possesses \textit{a priori} knowledge of the latency map $\latency_e(\cdot)$ for each edge $e \in \edges$, 
    obtained through analytical models, simulations, or 
    data.
    Below, we use the notation $\latencyVec: \R^{|\edges|} \ra \R^{|\edges|}$, defined by $\latencyVec(\flowVec) := \big( \latency_e(\flow_e): e \in \edges \big)$ for each $\flowVec \in \R^{|\edges|}$.
    \item 
    The toll designer possesses a set $\mc{D} = \{(\flowVec^\paren{i}, \latencyDataVec^\paren{i}) \in \R^{2|\edges|}: i \in [N] \}$ of observed flow-latency data pairs, collected across network edges $\edges$.
    \item The toll designer is aware that
    each flow-latency data pair in $\mc{D}$ arises from the following probabilistic model:
    \begin{align}
        \latencyDataVec^\paren{i} &= \latencyVec(\flowVec^\paren{i}) + \disturbVec^\paren{i}.
    \end{align}
    Above, $\{\disturbVec^\paren{i}: i \in [N] \} := \mc{D}_\disturb$ are 
    i.i.d. samples from a ground truth disturbance distribution $\disturbDistr_{\text{true}}$ unknown to the toll designer, who constructs an empirical distribution $\disturbDistrNominal$ from $\mc{D}_\disturb$ to approximate $\disturbDistr_{\text{true}}$.
    We henceforth refer to $\disturbDistrNominal$ as the \textit{nominal} disturbance distribution.
\end{enumerate}
The empirical mean and covariance associated with $\disturbDistrNominal$, which we denote by $\meanDisturbNominal$ and $\covDisturbNominal$, can then be computed as follows:
\begin{align}
    \meanDisturbNominal &= \frac{1}{N}\sum_{i} \big[ \latencyDataVec^i - \latencyVec(\flow^i) \big] \in \R^{|\edges|}, \\
    \covDisturbNominal &= \frac{1}{N}\sum_{i} (\latencyDataVec^i - \meanDisturbNominal)(\latencyDataVec^i - \meanDisturbNominal)^\top \in \R^{|\edges|\times |\edges|}.
\end{align}
To facilitate traffic flow prediction or toll design, the empirical distribution $\disturbDistrNominal$ is often approximated via best fit to a parameterized class of distributions, e.g., Gaussians.
This approach provides an analytical description of $\disturbDistrNominal$, but does not account for any distributional misspecification arising from a pre-chosen model class, or finite sample uncertainty induced by the limited data in $\mc{D}$. 
Thus, instead of selecting a single distribution from a parametric class,
we instead leverage the \textit{2-Wasserstein metric} $W_2(\cdot, \cdot)$ over probability distributions to consider the \textit{non-parametric} collection of all distributions in the vicinity of $\disturbDistrNominal$. 
Concretely, consider the following $\epsilon$-\textit{Wasserstein ball} of \textit{all} probability distributions $\mu_{\disturb} \in \mc{P}(\R^{|\edges|})$ 
that are $\epsilon$-close to $\disturbDistrNominal$ in the Wasserstein sense:
\begin{equation}\label{eqn:Wasser_ball}
  \latencySlopeMatrix_{\disturbDistrDeviation}(\disturbDistrNominal) = \{\mu_{\disturb} \in \mc{P}(\R^{|\edges|}) | W_2(\mu_{\disturb}, \disturbDistrNominal) \leq \disturbDistrDeviation\}. 
\end{equation}
The Wasserstein ball $\latencySlopeMatrix_{\disturbDistrDeviation}(\disturbDistrNominal)$ extends the empirical distribution $\disturbDistrNominal$ to a set of probability distributions that promote robustness against the following uncertainty sources: 
\begin{enumerate}
    \item \textbf{Finite sample uncertainty}: $B_\epsilon(\disturbDistrNominal)$ accounts for errors induced by the limited data in $\mc{D}$;
    \item \textbf{Model misspecification}:  
    $B_\epsilon(\disturbDistrNominal)$ includes all distributions near the nominal disturbance distribution, including 
    skewed, or heavy-tail distributions. 
    \item \textbf{Joint dependence across edges}: $B_\epsilon(\disturbDistrNominal)$ captures disturbance correlations across network edges. 
    \item \textbf{Latency model uncertainty}:
    If the worst-case mismatch between the latency model $\latency_e(\cdot)$ and the real-world system latency is uniformly bounded, this mismatch may be captured by suitably inflating $B_\epsilon(\disturbDistrNominal)$.
    
\end{enumerate}
In its current form, $B_\epsilon(\disturbDistrNominal)$ is not computationally tractable
to compute.
Instead, we consider its outer set approximation by the Gelbrich ball, given by:
\begin{align} \nonumber
    &\gelbrichBall\disturbDistrDeviation(\meanDisturbNominal, \covDisturbNominal) \\ \nonumber
    = &\{(\meanDisturb, \covDisturb) | \Vert \meanDisturb - \hat \meanDisturb \Vert_2^2 + tr(\covDisturb + \hat \covDisturb -2 (\hat \covDisturb^{1/2} \covDisturb \hat \covDisturb^{1/2})^{1/2} )\leq \disturbDistrDeviation^2\}.
\end{align}
From \cite[Thm.13]{Kuhn2019WassersteinDistributionallyRobustOptimizationTheoryApplicationsML}, 
$B_\epsilon(\disturbDistrNominal) \subseteq \gelbrichBall\disturbDistrDeviation(\meanDisturbNominal,\covDisturbNominal)$
with set equality if $\disturbDistrNominal$ is a multi-variable Gaussian distribution on $\R^{|\edges|}$. Additionally, the set $\gelbrichBall\disturbDistrDeviation(\meanDisturbNominal, \covDisturbNominal)$ 
is convex over $\R^{|\edges|} \times \R^{|\edges| \times |\edges|}$, and thus facilitates the construction of convex programs that encode distributional robustness
\cite{Taskesen2023DistributionallyRobustLQR}. Thus, to capture distributional uncertainty,
we perform toll optimization over the worst-case disturbance distribution $\disturbDistr$ in
the Gelbrich ball  $\gelbrichBall\disturbDistrDeviation(\meanDisturbNominal, \covDisturbNominal)$,
as described below in \eqref{Eqn: (Ideal) Problem statement} in Sec. \ref{subsec: A Mathematical Formulation for Distributionally Robust Toll Design}. 

\section{Preliminaries and Simplifying Assumptions}
\label{sec: Preliminaries and Simplifying Assumptions}

In Sec. \ref{subsec: A Mathematical Formulation for Distributionally Robust Toll Design}, we formulate the distributionally robust toll design problem studied in our work, and describe the challenges and technical issues inherent in the direct computation of distributionally robust tolls. 
Motivated by these challenges,
in Sec. \ref{subsec: Key Assumptions and Consequences}, we present a set of practically relevant conditions on the latency functions (Assumption \ref{Assump: Linear Latency Functions}) over network edges $\edges$.
We also present a condition ensuring the existence of tolls which guarantee that all edges in $\network$ are utilized at equilibrium despite worst-case disturbances (Assumption \ref{Assump: Toll Bounds for Ensuring Full Network Utilization}).
Finally, in Sec. \ref{subsec: Problem Statement}, we utilize Assumptions \ref{Assump: Linear Latency Functions}-\ref{Assump: Toll Bounds for Ensuring Full Network Utilization} to concretely formulate 
our problem statement
within our distributionally robust toll design framework.

\subsection{A Mathematical Formulation for Distributionally Robust Toll Design}
\label{subsec: A Mathematical Formulation for Distributionally Robust Toll Design}

We aim
to compute tolls that minimize the worst-case expected system-wide latency
at equilibrium,
when the disturbance $\disturbVec$ is generated adversarially from a distribution $\disturbDistr$ within an $\disturbDistrDeviation$-Gelbrich ball of a nominal distribution $\disturbDistrNominal$:
\begin{align} \label{Eqn: (Ideal) Problem statement}
    \text{arg}\min_{\tollVec \in \R_{\geq 0}^{|\edges|}} \max_{\disturbDistr \in \gelbrichBall_\disturbDistrDeviation(\disturbDistrNominal)} \E_{\disturbVec \sim \disturbDistr}[\latencyWorstEq(\disturbVec, \tollVec)].
\end{align}
When $\disturbDistrDeviation = 0$, the solution to the distributionally robust toll design problem is equivalent to the nominal toll design problem,
which does not account for error in modeling $\disturbVec$.

A traffic controller who aims to compute distributionally robust tolls by solving \eqref{Eqn: (Ideal) Problem statement} faces several challenges and limitations. First, a congestion game may in general yield many Nash equilibrium flows, thus rendering 
the worst-case system-wide latency $\latencyWorstEq(\cdot, \cdot)$
in \eqref{Eqn: (Ideal) Problem statement}
challenging to compute and analyze.
Second, it is unclear whether the outer optimization objective in \eqref{Eqn: (Ideal) Problem statement}, i.e., $\max_{\disturbDistr \in \gelbrichBall_\disturbDistrDeviation(\disturbDistrNominal)} \E_{\disturbVec \in \disturbDistr}\big[ \latencyWorstEq(\disturbVec, \tollVec) \big]$,
would be a convex function of the toll 
$\tollVec$.
Finally, the equilibrium flow which globally optimizes \eqref{Eqn: (Ideal) Problem statement}, even when tractable to compute, may not fully utilize all network edges, a key concern for toll designers who aim to promote the efficient use of traffic infrastructure through congestion pricing.
In response to these challenges, we present a set of practically relevant assumptions which aid in the construction of our problem statement \eqref{Eqn: Problem Statement}, as presented in Sec. \ref{subsec: Problem Statement}, from \eqref{Eqn: (Ideal) Problem statement}.

\subsection{Key Assumptions and Consequences}
\label{subsec: Key Assumptions and Consequences}


To formulate a more tractable and practically relevant version of the distributionally robust toll design objective \eqref{Eqn: (Ideal) Problem statement}, we present the following assumptions in our theoretical derivations.
First, we assume that the traffic network $\network$ considered in our work has strictly increasing linear latency functions\footnote{Our formulation extends readily to settings with affine latency functions by shifting the mean of the disturbance distribution.}. We note that the use of affine or piecewise affine latency functions in traffic network models, to facilitate theoretical or empirical analysis, is common in the congestion pricing literature \cite{Gollapudi2023OnlineLearningforTrafficNavigationinCongestedNetworks, YueFerguson2021IncentiveDesignforCongestionGameswithUnincentivizableUsers, BrownMarden2017RobustnessofMarginalCostTaxesinAffineCongestionGames}.


\begin{assumption}[\textbf{Linear Latency Functions}]
\label{Assump: Linear Latency Functions}
For each $e \in \edges$, there exists $\latencySlope_e > 0$ such that $\latency_e(\flow_e) = \latencySlope_e \flow_e$ for all $\flow_e \geq 0$.
\end{assumption}

Although Assumption \ref{Assump: Linear Latency Functions} requires each latency map $\latency_e(\cdot)$ to be linear rather than affine, our results naturally extend to the affine latency setting by suitably adjusting the mean of the disturbance distribution $\disturbDistr$.


The stipulation $\latencySlope_e > 0$ in Assumption \ref{Assump: Linear Latency Functions} renders the Rosenthal potential function \eqref{Eqn: Rosenthal Potential Function} strongly convex. Subsequently, \eqref{Eqn: Nash Computation via Rosenthal Potential Function} then implies that the game $\gameCongestion(\disturbVec, \tollVec)$ always admits a unique Nash equilibrium flow $\flowVecNash(\disturbVec, \tollVec)$.
For brevity, we write $\latencySlopeVec := (\latencySlope_e: e \in \edges) \in \R^{|\edges|}$ and $\latencySlopeMatrix := \diag\{\latencySlopeVec\}$.

Below, in Prop. \ref{Prop: KKT matrix with B and R is invertible, and has psd upper left block} and \ref{Prop: Null space of KKT matrix psd upper left block}, we present some useful relations that between $\latencySlopeMatrix$ and $\routingMatrix$ that will facilitate the proofs of our main results (Thms. \ref{Thm: (g) Characterization of Equilibrium System Latency} and \ref{Thm: Optimal Toll for Distributional Robustness}).

\begin{proposition} \label{Prop: KKT matrix with B and R is invertible, and has psd upper left block}
Let $\slantProjMatrix := - \latencySlopeMatrix^{-1} \routingMatrix^\top (\routingMatrix \latencySlopeMatrix^{-1} \routingMatrix^\top)^{-1} \routingMatrix \latencySlopeMatrix^{-1} \allowbreak + \latencySlopeMatrix^{-1}$. Then:
\begin{align} \label{Eqn: Inverse of B, R top; R, O}
    &\begin{bmatrix}
        \latencySlopeMatrix & \routingMatrix^\top \\
        \routingMatrix & O
    \end{bmatrix}^{-1} \\ \nonumber
    = \ &\begin{bmatrix}
        \slantProjMatrix & \latencySlopeMatrix^{-1} \routingMatrix^\top (\routingMatrix \latencySlopeMatrix^{-1} \routingMatrix^\top)^{-1} \\
        (\routingMatrix \latencySlopeMatrix^{-1} \routingMatrix^\top)^{-1} \routingMatrix \latencySlopeMatrix^{-1} & - (\routingMatrix \latencySlopeMatrix^{-1} \routingMatrix^\top)^{-1}
    \end{bmatrix}.
\end{align}
Moreover, $\slantProjMatrix$ is symmetric positive semidefinite.
\end{proposition}

\begin{proof}
\eqref{Eqn: Inverse of B, R top; R, O} follows straightforwardly from algebraic manipulations; we note that $\routingMatrix \latencySlopeMatrix^{-1} \routingMatrix^\top$ is invertible, since $\latencySlopeMatrix = \diag\{\latencySlopeVec\}$ is 
invertible
under Assumption \ref{Assump: Linear Latency Functions}
and $\routingMatrix$ has full row rank \cite[Prop. 2.8]{Chiu2026TechnicalNoteTechnicalNoteGraphTheoreticPropertiesofTrafficNetworks}.
To prove that $\slantProjMatrix$ is symmetric positive semidefinite, we note that:
\begin{align} \label{Eqn: Projection Interpretation of upper left block of B, R; top R, O inverse}
    \slantProjMatrix = \ &\latencySlopeMatrix^{-1/2} \big( I - \latencySlopeMatrix^{-1/2} \routingMatrix^\top (\routingMatrix \latencySlopeMatrix^{-1} \routingMatrix^\top)^{-1} \routingMatrix \latencySlopeMatrix^{-1/2} \big) \latencySlopeMatrix^{-1/2}.
\end{align}
Since $\latencySlopeMatrix$ is symmetric positive definite and $\routingMatrix$ has full row rank, $\latencySlopeMatrix^{-1/2}\routingMatrix^\top$ has full column rank. 
It
is well-known that, given a matrix $X \in \R^{n \times p}$ with full column rank, the orthogonal projection onto $\rangespace(X)$, i.e., the range space of $X$, can be represented by the matrix $X(X^\top X)^{-1} X$. Thus, the matrix $\latencySlopeMatrix^{-1/2} \routingMatrix^\top (\routingMatrix \latencySlopeMatrix^{-1} \routingMatrix^\top)^{-1} \routingMatrix \latencySlopeMatrix^{-1/2}$
in \eqref{Eqn: Projection Interpretation of upper left block of B, R; top R, O inverse} represents the orthogonal projection onto $\rangespace(\latencySlopeMatrix^{-1/2} \routingMatrix^\top)$. 
Thus,
$I - \latencySlopeMatrix^{-1/2} \routingMatrix^\top (\routingMatrix \latencySlopeMatrix^{-1} \routingMatrix^\top)^{-1} \routingMatrix \latencySlopeMatrix^{-1/2}$ represents the orthogonal projection onto $\rangespace(\latencySlopeMatrix^{-1/2} \routingMatrix^\top)^\perp$, and is thus symmetric positive semidefinite. Thus, the matrix \eqref{Eqn: Projection Interpretation of upper left block of B, R; top R, O inverse} is likewise symmetric positive semidefinite, as desired.
\end{proof}

\begin{proposition}
\label{Prop: Null space of KKT matrix psd upper left block}
$\nullspace(\slantProjMatrix) = \rangespace(\routingMatrix^\top)$.
\end{proposition}

\begin{proof}
To prove $\nullspace(\slantProjMatrix) \subseteq \rangespace(\routingMatrix^\top)$, let $u \in \nullspace(\slantProjMatrix)$. Then:
\begin{align} \nonumber
    0 &= \latencySlopeMatrix \slantProjMatrix u \\ \nonumber
    &= \latencySlopeMatrix (- \latencySlopeMatrix^{-1} \routingMatrix^\top (\routingMatrix \latencySlopeMatrix^{-1} \routingMatrix^\top)^{-1} \routingMatrix \latencySlopeMatrix^{-1} + \latencySlopeMatrix^{-1}) u \\ \nonumber
    &= - \routingMatrix^\top (\routingMatrix \latencySlopeMatrix^{-1} \routingMatrix^\top)^{-1} \routingMatrix \latencySlopeMatrix^{-1} u + u
\end{align}
Rearranging gives $u = \routingMatrix^\top (\routingMatrix \latencySlopeMatrix^{-1} \routingMatrix^\top)^{-1} \routingMatrix \latencySlopeMatrix^{-1} u \in \rangespace(\routingMatrix^\top)$. Thus, $\nullspace(\slantProjMatrix) \subseteq \rangespace(\routingMatrix^\top)$, as desired.

To prove $\rangespace(\routingMatrix^\top) \subseteq \nullspace(\slantProjMatrix)$, let $u \in \rangespace(\routingMatrix^\top)$. Then there exists some $v \in \R^{|\nodes|-1}$ such that $u = \routingMatrix^\top v$. Thus:
\begin{align} \nonumber
    \slantProjMatrix u &= \slantProjMatrix \routingMatrix^\top v \\ \nonumber
    &= (- \latencySlopeMatrix^{-1} \routingMatrix^\top (\routingMatrix \latencySlopeMatrix^{-1} \routingMatrix^\top)^{-1} \routingMatrix \latencySlopeMatrix^{-1} \allowbreak + \latencySlopeMatrix^{-1}) \routingMatrix^\top v \\ \nonumber
    &= - \latencySlopeMatrix^{-1} \routingMatrix^\top v
\end{align}
so $u \in \nullspace(\slantProjMatrix)$. Therefore, $\rangespace(\routingMatrix^\top) \subseteq \nullspace(\slantProjMatrix)$, as desired.
\end{proof}

Next, we present Assumption \ref{Assump: Toll Bounds for Ensuring Full Network Utilization}, which constrains disturbances and tolls to ensure all network edges $\edges$ are utilized at equilibrium regardless of the disturbance realization $\disturbVec$, a condition we term \textit{robust full network utilization}. 

To formally present Assumption \ref{Assump: Toll Bounds for Ensuring Full Network Utilization}, we require the following notation. 
For each $\disturbDistrDeviation, \disturbRange \geq 0$, we define 
the following subsets
toll sets
$\tollSet(\disturbDistrDeviation, \disturbRange), \tollSetInterior(\disturbDistrDeviation, \disturbRange) \subseteq \R^{|\edges|}$, 
which will facilitate the enforcement of robust full network utilization:
\begin{align} \label{Eqn: (Full) Toll Constraint Set}
    \tollSet(\disturbDistrDeviation, \disturbRange) &:= \{ \tollVec \in \R_{\geq 0}^{|\edges|}: \slantProjMatrix \tollVec \leq - \Vert \slantProjMatrix \Vert_2 (\disturbDistrDeviation + \disturbRange) \onesVector_{|\edges|} \\ \nonumber
    &\hspace{2.5cm} - \slantProjMatrix \meanDisturbNominal + \latencySlopeMatrix^{-1} \routingMatrix^\top (\routingMatrix \latencySlopeMatrix^{-1} \routingMatrix^\top)^{-1} \demandVec \}, \\ \label{Eqn: (Interior) Toll Constraint Set}
    \tollSetInterior(\disturbDistrDeviation, \disturbRange) &:= \{ \tollVec \in \R_{\geq 0}^{|\edges|}: \slantProjMatrix \tollVec < - \Vert \slantProjMatrix \Vert_2 (\disturbDistrDeviation + \disturbRange) \onesVector_{|\edges|} \\ \nonumber
    &\hspace{2.5cm} - \slantProjMatrix \meanDisturbNominal + \latencySlopeMatrix^{-1} \routingMatrix^\top (\routingMatrix \latencySlopeMatrix^{-1} \routingMatrix^\top)^{-1} \demandVec \}.
\end{align}
We note a monotonicity property that immediately follows from the definitions of $\tollSet$ and $\tollSetInterior$: For any $\disturbRange \geq 0$ and any $\disturbDistrDeviation_1, \disturbDistrDeviation_2 \geq 0$ such that $\disturbDistrDeviation_2 \geq \disturbDistrDeviation_1$, we have $\tollSet(\disturbDistrDeviation_2, \disturbRange) \subseteq \tollSet(\disturbDistrDeviation_1, \disturbRange)$ and $\tollSetInterior(\disturbDistrDeviation_2, \disturbRange) \subseteq \tollSetInterior(\disturbDistrDeviation_1, \disturbRange)$.

We now introduce Assumption \ref{Assump: Toll Bounds for Ensuring Full Network Utilization}.

\begin{assumption}(\textbf{Toll Bounds for Ensuring Full Network Utilization})
\label{Assump: Toll Bounds for Ensuring Full Network Utilization}
The nominal disturbance distribution $\disturbDistrNominal$, maximum spread $\delta$ of the disturbance distribution,
and the flow demand $\demand_s$ 
satisfy $\tollSetInterior(0, \disturbRange) \ne \emptyset$, i.e., there exists $\hat\tollVec \in \R_{\geq 0}^{|\edges|}$ such that $\slantProjMatrix \hat\tollVec < - \Vert \slantProjMatrix \Vert_2 \disturbRange \onesVector_{|\edges|} - \slantProjMatrix \meanDisturbNominal + \latencySlopeMatrix^{-1} \routingMatrix^\top (\routingMatrix \latencySlopeMatrix^{-1} \routingMatrix^\top)^{-1} \demandVec$.


\end{assumption}

We now characterize the maximum deviation $\disturbDistrDeviation > 0$ of the mean of the disturbance distribution $\disturbDistr$
under which the condition of robust full network utilization can be enforced.

\begin{lemma} \label{Lemma: Max Distributional Deviation}
Suppose Assumption \ref{Assump: Linear Latency Functions} holds, and Assumption \ref{Assump: Toll Bounds for Ensuring Full Network Utilization} holds for the nominal distribution $\disturbDistrNominal$, distribution spread $\disturbRange > 0$, and flow demand $\demand_s > 0$. Then the following linear program (LP):
\begin{subequations} \label{Eqn: LP for epsilon max}
\begin{align} \label{Eqn: LP for epsilon max, Objective}
    \disturbDistrDeviationMax := \max_{\disturbDistrDeviation, \tollVec} \hspace{5mm} &\disturbDistrDeviation \\ \label{Eqn: LP for epsilon max, Toll Bounds Constraint}
    \text{s.t.} \hspace{5mm} &\tollVec \in \tollSet(\disturbDistrDeviation, \disturbRange)
\end{align}
\end{subequations}
is feasible, with $\disturbDistrDeviationMax > 0$. Moreover, for any $\disturbDistrDeviation \geq 0$, we have $\tollSet(\disturbDistrDeviation, \disturbRange) \ne \emptyset$ if and only if $\disturbDistrDeviation \in [0, \disturbDistrDeviationMax]$, and $\tollSetInterior(\disturbDistrDeviation, \disturbRange) \ne \emptyset$ if and only if $\disturbDistrDeviation \in [0, \disturbDistrDeviationMax)$.
\end{lemma}

\begin{proof}
First, we prove that the LP \eqref{Eqn: LP for epsilon max} is feasible, with $\disturbDistrDeviationMax > 0$.
Since Assumption \ref{Assump: Toll Bounds for Ensuring Full Network Utilization} holds, there exists some $\hat\tollVec \in \R_{\geq 0}^{|\edges|}$ such that:
\begin{align}
    \slantProjMatrix \hat\tollVec &< - \Vert \slantProjMatrix \Vert_2 \disturbRange \onesVector_{|\edges|} - \slantProjMatrix \meanDisturbNominal + \latencySlopeMatrix^{-1} \routingMatrix^\top (\routingMatrix \latencySlopeMatrix^{-1} \routingMatrix^\top)^{-1} \demandVec.
\end{align}
If $\slantProjMatrix = O$, then $\tollSet(\disturbDistrDeviation, \disturbRange)$ becomes independent of $\disturbDistrDeviation$ and $\disturbRange$, and so Assumption \ref{Assump: Toll Bounds for Ensuring Full Network Utilization} (i.e., $\tollSetInterior(0, \disturbRange) \ne \emptyset$) guarantees that \eqref{Eqn: LP for epsilon max, Toll Bounds Constraint} holds for any $\disturbDistrDeviation > 0$; thus, $\disturbDistrDeviationMax = +\infty$.
Otherwise, $\slantProjMatrix \ne O$, so $\Vert \slantProjMatrix \Vert_2 > 0$. We define:
\begin{align}
    \hat\disturbDistrDeviation &:= \frac{1}{\Vert \slantProjMatrix \Vert_2} \min_{k \in [|\edges|]} \big\{ \big( -\slantProjMatrix \hat\tollVec - \Vert \slantProjMatrix \Vert_2 \disturbRange \onesVector_{|\edges|} - \slantProjMatrix \meanDisturbNominal \\ \nonumber
    &\hspace{2.5cm} + \latencySlopeMatrix^{-1} \routingMatrix^\top (\routingMatrix \latencySlopeMatrix^{-1} \routingMatrix^\top)^{-1} \demandVec \big)_k \big\} > 0,
\end{align}
where $(\cdot)_k$ denotes the $k$-th component of a given input vector.
Then:
\begin{align} \nonumber
    &-\slantProjMatrix \hat\tollVec - \Vert \slantProjMatrix \Vert_2 (\disturbRange + \hat\disturbDistrDeviation) \onesVector_{|\edges|} - \slantProjMatrix \meanDisturbNominal \\ \nonumber
    &\hspace{1cm} + \latencySlopeMatrix^{-1} \routingMatrix^\top (\routingMatrix \latencySlopeMatrix^{-1} \routingMatrix^\top)^{-1} \demandVec \\ \nonumber
    = \ &-\hat\disturbDistrDeviation \Vert \slantProjMatrix \Vert_2 \onesVector_{|\edges|} + \big[ -\slantProjMatrix \hat\tollVec - \Vert \slantProjMatrix \Vert_2 (\disturbRange + \hat\disturbDistrDeviation) \onesVector_{|\edges|} \\ \nonumber
    &\hspace{1cm} - \slantProjMatrix \meanDisturbNominal + \latencySlopeMatrix^{-1} \routingMatrix^\top (\routingMatrix \latencySlopeMatrix^{-1} \routingMatrix^\top)^{-1} \demandVec \big] \\ \nonumber
    \geq \ &0.
\end{align}
Moreover, $\hat{\tollVec} \in \R_{\geq 0}^{|\edges|}$ and $\hat \disturbDistrDeviation \geq 0$ by construction, so $(\hat\tollVec, \hat\disturbDistrDeviation)$ is a feasible point of the LP \eqref{Eqn: LP for epsilon max}.
Thus, $\disturbDistrDeviationMax \geq \hat \disturbDistrDeviation > 0$.

Next, we prove that for any $\disturbDistrDeviation \geq 0$, we have $\tollSet(\disturbDistrDeviation, \disturbRange) \ne \emptyset$ if and only if $\disturbDistrDeviation \in [0, \disturbDistrDeviationMax]$. 
First, suppose $\tollSet(\disturbDistrDeviation, \disturbRange) \ne \emptyset$ for some $\disturbDistrDeviation \geq 0$, and take $\tollVec \in \tollSet(\disturbDistrDeviation, \disturbRange)$. Then $(\disturbDistrDeviation, \disturbRange)$ is feasible with respect to \eqref{Eqn: LP for epsilon max}, so $\disturbDistrDeviation \leq \disturbDistrDeviationMax$ by definition of $\disturbDistrDeviationMax$. Thus, $\disturbDistrDeviation \in [0, \disturbDistrDeviationMax]$, as desired. 
Conversely, suppose $\disturbDistrDeviation \in [0, \disturbDistrDeviationMax]$, then there must exist some $\tollVec \in \tollSet(\disturbDistrDeviationMax, \disturbRange)$ (otherwise, $\disturbDistrDeviationMax$ would not be feasible with respect to \eqref{Eqn: LP for epsilon max}, a contradiction). 
Since $\disturbDistrDeviation \in [0, \disturbDistrDeviationMax]$, by the monotonicity property of $\tollSet$, as remarked immediately after its definition, we have $\tollVec \in \tollSet(\disturbDistrDeviationMax, \disturbRange) \subseteq \tollSet(\disturbDistrDeviation, \disturbRange)$. Thus, $\tollSet(\disturbDistrDeviation, \disturbRange) \ne \emptyset$.

\end{proof}

As revealed in the proof of Thm. \ref{Thm: (g) Characterization of Equilibrium System Latency}, Assumption \ref{Assump: Toll Bounds for Ensuring Full Network Utilization} provides a sufficient condition on the toll $\tollVec$ to guarantee that the equilibrium flows corresponding to the deployment of $\tollVec$ have non-zero occupancy over each edge of the network.

\subsection{Problem Statement}
\label{subsec: Problem Statement}

We
now present our problem statement, which modifies the robust optimization program \eqref{Eqn: (Ideal) Problem statement} to seek a toll $\tollVec \in \R_{\geq 0}^{|\edges|}$ which minimizes the 
equilibrium system-wide latency, while robustly guaranteeing that the traffic network is fully utilized.
Concretely, 
suppose Assumption \ref{Assump: Linear Latency Functions} holds, and Assumption \ref{Assump: Toll Bounds for Ensuring Full Network Utilization} holds for the nominal distribution $\disturbDistrNominal$, distribution spread $\disturbRange > 0$, and flow demand $\demand_s > 0$. 
Let $\disturbDistrDeviationMax$ be as given in Lemma \ref{Lemma: Max Distributional Deviation}. 
Further, let $\M(\disturbRange)$ denote the set of disturbance distributions over $\R^n$ that, with probability 1, are supported in a 2-norm ball of radius $\disturbRange$ from its mean:
\begin{align} \label{Eqn: Local Disturbance Distributions}
    \M(\disturbRange) &:= \{\disturbDistr \in \mc{P}(\R^n): \Prob_{\disturbVec \sim \disturbDistr}\big( \Vert \disturbVec - \E[\disturbVec] \Vert_2 \leq \disturbRange \big) = 1. \}.
\end{align}
Then, for any $\disturbDistrDeviation \in [0, \disturbDistrDeviationMax]$, and $\tollVec \in \tollSet(\disturbDistrDeviation, \disturbRange)$, we aim to compute the set of \textit{distributionally robust tolls}, denoted below by $\optTollSet(\disturbDistrDeviation)$, which minimizes the 
system-wide latency at equilibrium over all possible disturbance random variables $\disturbVec \sim \disturbDistr$, with $\disturbDistr \in \gelbrichBall_\disturbDistrDeviation(\disturbDistrNominal) \cap \M(\disturbRange)$, while robustly guaranteeing full network utilization:
\begin{align} \label{Eqn: Problem Statement}
    \optTollSet(\disturbDistrDeviation) := \text{arg}\min_{\tollVec \in \tollSet(\disturbDistrDeviation, \disturbRange)} \max_{\disturbDistr \in \gelbrichBall_\disturbDistrDeviation(\disturbDistrNominal) \cap \M(\disturbRange)} \E_{\disturbVec \sim \disturbDistr} \big[ \latencyWorstEq(\disturbVec, \tollVec) \big].
\end{align}

\section{Main Results}
\label{sec: Main Results}


We describe the congestion game model and equilibrium traffic flow concepts underlying our main theoretical contributions (Sec. \ref{subsec: Characterization of Equilibrium Flows and Latency}). 
We then present a convex programming-based characterization of distributionally robust tolls, which serves as the main contribution of our work (Sec. \ref{subsec: Characterization of Distributionally Robust Optimal Tolls}).

\subsection{Characterization of Equilibrium Flows and Latency}
\label{subsec: Characterization of Equilibrium Flows and Latency}

\looseness=-1
First, we prove that under Assumptions \ref{Assump: Linear Latency Functions}-\ref{Assump: Toll Bounds for Ensuring Full Network Utilization} and appropriate bounds on the toll $\tollVec$,  the Nash equilibrium flow $\flowVecNash(\disturbVec, \tollVec)$ is an affine function of $(\disturbVec, \tollVec)$ with probability 1.


\begin{lemma}(\textbf{Characterization of Nash Equilibrium Flows})
\label{Lemma: (f NE) Characterization of Nash Equilibrium Flows}
Suppose Assumption \ref{Assump: Linear Latency Functions} holds.
Given $\disturbRange, \demand > 0$ and $\disturbDistrNominal$ satisfying Assumption \ref{Assump: Toll Bounds for Ensuring Full Network Utilization}, for any $\tollVec \in \tollSet(\disturbDistrDeviation, \disturbRange)$ and $\disturbDistr \in \gelbrichBall\disturbDistrDeviation(\disturbDistrNominal) \bigcap \M(\disturbRange)$, the Nash equilibrium flow (see \eqref{Eqn: Nash Equilibrium (path) flow, Def} and the related discussion) corresponding to $\tollVec$ and disturbance random variable $\disturbVec \sim \disturbDistr$, 
i.e., $\flowVecNash(\disturbVec, \tollVec)$, satisfies:
\begin{align} \label{Eqn: Nash equilibrium flow as affine function of disturbance and toll}
    \flowVecNash(\disturbVec, \tollVec) = - \slantProjMatrix(\disturbVec + \tollVec) + \latencySlopeMatrix^{-1}\routingMatrix^\top (\routingMatrix \latencySlopeMatrix^{-1} \routingMatrix^\top)^{-1} \demandVec)
\end{align}
with probability 1.
If we additionally have $\tollVec \in \tollSetInterior(\disturbDistrDeviation, \disturbRange)$, then $\flowVecNash(\disturbVec, \tollVec) > 0$ with probability 1, i.e., the network is robustly fully utilized.
\end{lemma}

\begin{proof}
Let $\eventDisburbanceBound$ denote the event in which $\Vert \disturbVec - \E_{\disturbVec \sim \disturbDistr}[\disturbVec] \Vert_2 \leq \disturbRange$ when $\disturbVec$ is drawn from $\disturbDistr$. Since $\disturbDistr \in \M(\disturbRange)$, we have $\mc{P}(\eventDisburbanceBound) = 1$. Thus, it suffices to prove that \eqref{Eqn: Nash equilibrium flow as affine function of disturbance and toll} holds whenever the event $\eventDisburbanceBound$ occurs.

Under Assumption \ref{Assump: Linear Latency Functions}, the Nash equilibrium flow corresponding to the game $\gameCongestion$ is unique, i.e., $\setNashFlows(\disturbVec, \tollVec) = \{ \flowVecNash(\disturbVec, \tollVec) \}$.
Thus:
\begin{align} \nonumber
    \latencyWorstEq(\disturbVec, \tollVec) &:= \max_{\flowVec \in \setNashFlows(\disturbVec, \tollVec)} \sum_{e \in \edges} \flow_e \big[\latency_e(\flow_e) + \disturb_e \big] \\ \label{Eqn: g reformulated, general networks, linear latency functions}
    &= \sum_{e \in \edges} \latencySlope_e \big( \flowNash_e(\disturbVec, \tollVec) \big)^2 + \disturb_e \flowNash_e(\disturbVec, \tollVec).
\end{align}
Moreover, as described in \eqref{Eqn: Nash Computation via Rosenthal Potential Function}, $\flowVecNash(\disturbVec, \tollVec)$ is the (in this case, unique) minimizer of the Rosenthal potential function \eqref{Eqn: Rosenthal Potential Function}. Thus, we can write:
\begin{align} \nonumber
    \flowVecNash(\disturbVec, \tollVec) &= \arg \min_{\flowVec \in \setFeasibleFlows} \sum_{e \in \edges} \int_0^{\flow_e} \big( \latencySlope_e w + \disturb_e \big) dw + \toll_e \flow_e \\ \label{Eqn: f NE, general networks, linear latency functions}
    &= \arg \min_{\flowVec \in \setFeasibleFlows} \sum_{e \in \edges} \Big( \frac{1}{2} \latencySlope_e \flow_e^2 + (\disturb_e + \toll_e) \flow_e \Big)
\end{align}
Let $\dualIneq_e$ denote the dual variable associated with the non-negativity constraint $\flow_e \geq 0$, for each $e \in \edges$, and let $\dualEq_i$ denote the dual variable associated with the flow continuity constraint $\sum_{e \in \edges_i^+} \flow_e = \mathds{1}\{i=s\} + \sum_{e \in \edges_i^-} \flow_e$ at each node $i \in \nodes \backslash \{d\}$. For convenience, let $\dualIneqVec := (\dualIneq_e: e \in \edges)$ and $\dualEqVec := (\dualEq_i: i \in \nodes \backslash \{d\})$.
We write the Lagrangian of \eqref{Eqn: f NE, general networks, linear latency functions} by:
\begin{align} \nonumber
    &\lagrangian(\flowVec, \dualIneqVec, \dualEqVec) \\ \nonumber
    := \ &\sum_{e \in \edges} \Bigg[ \frac{1}{2} \latencySlope_e \flow_e^2 + (\disturb_e + \toll_e) \flow_e \Bigg] + \sum_{e \in \edges} \dualIneqVec^\top (-\flowVec) \\ \nonumber
    &\hspace{5mm} + \dualEqVec^\top (\routingMatrix \flowVec - \demandVec) 
\end{align} 
The KKT conditions for the optimal primal-dual variables $(\flowVec, \dualIneqVec, \dualEqVec) = (\flowVecNash(\disturbVec, \tollVec), \dualIneqVec^\star(\disturbVec, \tollVec), \dualEqVec^\star(\disturbVec, \tollVec))$ are thus:
\begin{subequations} \label{Eqn: KKT, general networks, linear latency functions}
\begin{align} \label{Eqn: KKT, Primal feasibility, flow continuity, general networks, linear latency functions}
    \demandVec &= \routingMatrix \flowVec, \\ 
    \label{Eqn: KKT, Primal feasibility, non-negativity, general networks, linear latency functions}
    0 &\leq \flow_e, \hspace{5mm} \forall e \in \edges, \\ \label{Eqn: KKT, Dual feasibility, general networks, linear latency functions}
    0 &\leq \dualIneq_e, \hspace{5mm} \forall e \in \edges, \\ \label{Eqn: KKT, Complementary Slackness, general networks, linear latency functions}
    0 &=  \dualIneq_e \flow_e, \hspace{5mm} \forall e \in \edges, \\ \label{Eqn: KKT, Stationarity, general networks, linear latency functions}
    0 &:= \frac{\partial \lagrangian}{\partial \flow_e} = \latencySlope_e \flowNash_e + (\disturb_e + \toll_e) - \dualIneq_e \\ \nonumber
    &\hspace{2cm} + (\routingMatrix^\top \dualEqVec)_e, \hspace{5mm} \forall e \in \edges.
\end{align}
\end{subequations}

We now claim that the primal-dual variable values defined below, i.e., $\hat \flowVec$, $\hat \dualIneqVec$, and $\hat \dualEqVec$, satisfy the KKT conditions \eqref{Eqn: KKT, general networks, linear latency functions}:
\begin{subequations} \label{Eqn: Optimal variables, general networks, linear latency functions}
\begin{align} \label{Eqn: f NE and nu star, general networks, linear latency functions}
    \begin{bmatrix}
        \hat\flowVec \\
        \hat\dualEqVec
    \end{bmatrix} &:= 
    \begin{bmatrix}
        \latencySlopeMatrix & \routingMatrix^\top \\
        \routingMatrix & O
    \end{bmatrix}^{-1}
    \begin{bmatrix}
        -\disturbVec - \tollVec \\ \demandVec
    \end{bmatrix}, \\ \label{Eqn: lambda star, general networks, linear latency functions}
    \hat\dualIneqVec &= 0.
\end{align}
\end{subequations}

First, \eqref{Eqn: KKT, Primal feasibility, flow continuity, general networks, linear latency functions} and \eqref{Eqn: KKT, Stationarity, general networks, linear latency functions} can be equivalently expressed as:
\begin{align}
    \begin{bmatrix}
        \latencySlopeMatrix & \routingMatrix^\top \\
        \routingMatrix & O
    \end{bmatrix}
    \begin{bmatrix}
        \flowVec \\ \dualEqVec
    \end{bmatrix}
    = \begin{bmatrix}
        -\disturbVec + \dualIneqVec - \tollVec \\ 
        \demandVec
    \end{bmatrix},
\end{align}
which is evidently satisfied when $(\flowVec, \dualEqVec, \dualIneqVec) = (\hat\flowVec, \hat\dualEqVec, \hat\dualIneqVec)$. Similarly, 
\eqref{Eqn: KKT, Dual feasibility, general networks, linear latency functions} and \eqref{Eqn: KKT, Complementary Slackness, general networks, linear latency functions} are satisfied when $\dualIneqVec = \hat\dualIneqVec$. Finally, to verify that \eqref{Eqn: KKT, Primal feasibility, non-negativity, general networks, linear latency functions} holds when $\flowVec = \hat\flowVec$, i.e., that $\hat\flowVec \geq 0$, we apply Prop. \ref{Prop: KKT matrix with B and R is invertible, and has psd upper left block} to 
reveal that:
\begin{align} \label{Eqn: (With hatted flow) Nash equilibrium flow as affine function of disturbance and toll}
    \hat\flowVec &= \slantProjMatrix(-\disturbVec - \tollVec) + \latencySlopeMatrix^{-1} \routingMatrix^\top (\routingMatrix \latencySlopeMatrix^{-1}\routingMatrix^\top)^{-1} \demandVec \\ \nonumber
    &= - \slantProjMatrix\tollVec - \slantProjMatrix(\disturbVec - \E_{\disturbVec \sim \disturbDistr}[\disturbVec]) - \slantProjMatrix(\E_{\disturbVec \sim \disturbDistr}[\disturbVec] - \meanDisturbNominal) \\ \nonumber
    &\hspace{1cm} - \slantProjMatrix\meanDisturbNominal + \latencySlopeMatrix^{-1} \routingMatrix^\top (\routingMatrix \latencySlopeMatrix^{-1}\routingMatrix^\top)^{-1} \demandVec \\ \label{Eqn: Applying disturbance bounds, in lemma on equilibrium flows' formula}
    &\geq - \slantProjMatrix\tollVec - \Vert \slantProjMatrix \Vert_2 (\disturbDistrDeviation + \disturbRange) \onesVector_{|\edges|} - \slantProjMatrix\meanDisturbNominal \\ \nonumber
    &\hspace{1cm} + \latencySlopeMatrix^{-1} \routingMatrix^\top (\routingMatrix \latencySlopeMatrix^{-1}\routingMatrix^\top)^{-1} \demandVec \\ \label{Eqn: Applying toll bounds, in lemma on equilibrium flows' formula}
    &\geq 0,
\end{align}
where \eqref{Eqn: Applying disturbance bounds, in lemma on equilibrium flows' formula} follows by applying the fact that $\disturbDistr \in \gelbrichBall\disturbDistrDeviation(\disturbDistrNominal) \bigcap \M(\disturbRange)$ and applying the Cauchy-Schwarz inequality, while \eqref{Eqn: Applying toll bounds, in lemma on equilibrium flows' formula} follows from the fact that $\tollVec \in \tollSet(\disturbDistrDeviation, \disturbRange)$.
Thus, to summarize, \eqref{Eqn: KKT, general networks, linear latency functions} holds when $(\flowVec, \dualIneqVec, \dualEqVec) = (\hat\flowVec, \hat\dualIneqVec, \hat\dualEqVec)$. 

Since $\flowVecNash(\disturbVec, \tollVec)$ is the unique minimizer of the convex program \eqref{Eqn: f NE, general networks, linear latency functions},
$\flowVecNash(\disturbVec, \tollVec)$ must be the unique flow for which there exists some $(\dualIneqVec, \dualEqVec)$ such that $(\flowVecNash(\disturbVec, \tollVec), \dualIneqVec, \dualEqVec)$ satisfies the KKT conditions in \eqref{Eqn: f NE, general networks, linear latency functions}. We thus conclude that:
\begin{align} \label{Eqn: Final equality, in lemma on equilibrium flows' formula}
    \flowVecNash(\disturbVec, \tollVec) &= - \slantProjMatrix(\disturbVec + \tollVec) + \latencySlopeMatrix^{-1}\routingMatrix^\top (\routingMatrix \latencySlopeMatrix^{-1} \routingMatrix^\top)^{-1} \demandVec),
\end{align}
where the equality \eqref{Eqn: Final equality, in lemma on equilibrium flows' formula} follows from \eqref{Eqn: (With hatted flow) Nash equilibrium flow as affine function of disturbance and toll}.

Finally, to prove that $\flowVecNash(\disturbVec, \tollVec) > 0$ with probability 1 when $\tollVec$ additionally satisfies $\tollVec \in \tollSetInterior(\disturbDistrDeviation, \disturbRange)$, we merely note that, in the above proof, \eqref{Eqn: Applying toll bounds, in lemma on equilibrium flows' formula} holds with strict inequality when $\tollVec \in \tollSetInterior(\disturbDistrDeviation, \disturbRange)$.
\end{proof}

Equipped with Lemma \ref{Lemma: (f NE) Characterization of Nash Equilibrium Flows}, we now aim to prove that under Assumptions \ref{Assump: Linear Latency Functions} and \ref{Assump: Toll Bounds for Ensuring Full Network Utilization}, the worst-case system latency realization under tolls $\tollVec \in \R_{\geq 0}^{|\edges|}$ and disturbances $\disturbVec \in \R_{\geq 0}^{|\edges|}$, given by $\latencyWorstEq(\disturbVec, \tollVec)$, is an affine function of $\disturbVec$ for any fixed $\tollVec$, and a convex quadratic function of $\tollVec$ for any fixed $\disturbVec$.

\begin{theorem}(\textbf{Characterization of Equilibrium System Latency})
\label{Thm: (g) Characterization of Equilibrium System Latency}
Under Assumption \ref{Assump: Linear Latency Functions},
given $\disturbRange, \demand > 0$ and $\disturbDistrNominal$ satisfying Assumption \ref{Assump: Toll Bounds for Ensuring Full Network Utilization}, for any $\tollVec \in \tollSet(\disturbDistrDeviation, \disturbRange)$ and $\disturbDistr \in \gelbrichBall\disturbDistrDeviation(\disturbDistrNominal) \bigcap \M(\disturbRange)$, the equilibrium system-wide latency is:
\begin{align} \label{Eqn: g, Final Expression, General Networks}
    \latencyWorstEq(\disturbVec, \tollVec) &=
    \begin{bmatrix}
        \tollVec^\top & \demandVec^\top
    \end{bmatrix}
    \begin{bmatrix}
        \latencySlopeMatrix & \routingMatrix^\top \\
        \routingMatrix & 0
    \end{bmatrix}^{-1}
    \begin{bmatrix}
        \disturbVec + \tollVec \\ -\demandVec 
    \end{bmatrix}
\end{align}
with probability 1.
\end{theorem}

\begin{proof}
From Lemma \ref{Lemma: (f NE) Characterization of Nash Equilibrium Flows}, we find that under the event $\eventDisburbanceBound := \big\{\Vert \disturbVec - \E_{\disturbVec \sim \disturbDistr}[\disturbVec] \Vert_2 \leq \disturbRange \big\}$ (which holds with probability 1 since $\disturbDistr \in \M(\disturbRange)$), \eqref{Eqn: Nash equilibrium flow as affine function of disturbance and toll} holds, i.e.,:
\begin{align} \nonumber
    \flowVecNash(\disturbVec, \tollVec) = - \slantProjMatrix(\disturbVec + \tollVec) + \latencySlopeMatrix^{-1}\routingMatrix^\top (\routingMatrix \latencySlopeMatrix^{-1} \routingMatrix^\top)^{-1} \demandVec).
\end{align}
Substituting \eqref{Eqn: Nash equilibrium flow as affine function of disturbance and toll} back into \eqref{Eqn: g reformulated, general networks, linear latency functions}, we find that:
\begin{align} \nonumber
    \latencyWorstEq(\disturbVec, \tollVec) &= \sum_{e \in \edges} \latencySlope_e \big( \flowNash_e(\disturbVec, \tollVec) \big)^2 + \disturb_e \flowNash_e(\disturbVec, \tollVec) \\ \nonumber
    &= \sum_{e \in \edges} \flowNash_e(\disturbVec, \tollVec) \cdot \Big( \latencySlope_e  \flowNash_e(\disturbVec, \tollVec) + \disturb_e \Big) \\ \label{Eqn: g reformulation, general networks, linear latency functions, via stationarity condition}
    &= \sum_{e \in \edges} \flowNash_e(\disturbVec, \tollVec) \cdot \big( -\toll_e - (\routingMatrix^\top \dualEqVec^\star)_e(\disturbVec, \tollVec) \big) \\ \nonumber
    &= - \tollVec^\top \flowVecNash(\disturbVec, \tollVec) - \flowVecNash(\disturbVec, \tollVec)^\top \routingMatrix^\top \dualEqVec^\star \\
    \label{Eqn: g reformulation, general networks, linear latency functions, via flow continuity}
    &= - \tollVec^\top \flowVecNash(\disturbVec, \tollVec) - \demandVec^\top \dualEqVec^\star \\ \label{Eqn: g reformulation, general networks, linear latency functions, via optimal f NE and nu star equations}
    &= \begin{bmatrix}
        -\tollVec^\top & - \demandVec^\top
    \end{bmatrix}
    \begin{bmatrix}
        \latencySlopeMatrix & \routingMatrix^\top \\
        \routingMatrix & 0
    \end{bmatrix}^{-1}
    \begin{bmatrix}
        -\disturbVec - \tollVec \\ \demandVec 
    \end{bmatrix} \\ \nonumber
    &= \begin{bmatrix}
        \tollVec^\top & \demandVec^\top
    \end{bmatrix}
    \begin{bmatrix}
        \latencySlopeMatrix & \routingMatrix^\top \\
        \routingMatrix & 0
    \end{bmatrix}^{-1}
    \begin{bmatrix}
        \disturbVec + \tollVec \\ -\demandVec 
    \end{bmatrix}
\end{align}
where \eqref{Eqn: g reformulation, general networks, linear latency functions, via stationarity condition} and \eqref{Eqn: g reformulation, general networks, linear latency functions, via flow continuity} respectively follow from \eqref{Eqn: KKT, Stationarity, general networks, linear latency functions} and \eqref{Eqn: KKT, Primal feasibility, flow continuity, general networks, linear latency functions}.

\end{proof}

\begin{remark}
Under the conditions imposed in the statement of Thm. \ref{Thm: (g) Characterization of Equilibrium System Latency}, $g$ is affine in $\disturbVec$ for fixed $\tollVec$ and quadratic in $\tollVec$ for fixed $\disturbVec$. 
Moreover, $g$ is convex in $\tollVec$ for fixed $\disturbVec$ if and only if the top left block of $[\latencySlopeMatrix, \routingMatrix^\top; \routingMatrix, O]^{-1}$ is symmetric positive semi-definite, which holds true due to Prop. \ref{Prop: KKT matrix with B and R is invertible, and has psd upper left block}.
\end{remark}

\subsection{Characterization of Distributionally Robust Optimal Tolls}
\label{subsec: Characterization of Distributionally Robust Optimal Tolls}

Armed with the expressions for $\flowVecNash$ and $g$ furnished by Prop. \ref{Lemma: (f NE) Characterization of Nash Equilibrium Flows} and Thm. \ref{Thm: (g) Characterization of Equilibrium System Latency}, we proceed to prove that the set $\optTollSet(\disturbDistrDeviation)$ of distributionally robust
tolls
defined by \eqref{Eqn: Problem Statement} can be computed via convex programming (Thm. \ref{Thm: Optimal Toll for Distributional Robustness}).


\begin{theorem} 
\label{Thm: Optimal Toll for Distributional Robustness}
Under
Assumptions \ref{Assump: Linear Latency Functions} and \ref{Assump: Toll Bounds for Ensuring Full Network Utilization}, 
the set of all distributionally robust optimal tolls $\optTollSet(\disturbDistrDeviation)$, as defined by \eqref{Eqn: Problem Statement}, is the minimizer set of the convex optimization problem \eqref{Eqn: Convex Program, Optimal Toll for Distributional Robustness}:
\begin{align} \label{Eqn: Convex Program, Optimal Toll for Distributional Robustness}
    &\min_{\tollVec \in \tollSet(\disturbDistrDeviation, \disturbRange)}. \hspace{2mm} \disturbDistrDeviation \big\Vert \slantProjMatrix \tollVec + \latencySlopeMatrix^{-1} \routingMatrix^\top (\routingMatrix \latencySlopeMatrix^{-1} \routingMatrix^\top)^{-1} \demandVec \big\Vert_2 \\ \nonumber
    &\hspace{1.8cm} 
    + \tollVec^\top \slantProjMatrix \tollVec + (\meanDisturbNominal)^\top \slantProjMatrix \tollVec,
\end{align}
Further, the corresponding worst-case disturbance distribution $\disturbDistr$ is obtained by shifting the nominal disturbance distribution $\disturbDistrNominal$ by:
\begin{align}
\label{Eqn: Worst-case Disturbance Distribution Shift}
    \E[\disturbDistr] = \meanDisturbNominal + \epsilon \cdot \frac{\slantProjMatrix \tollVec + \latencySlopeMatrix^{-1} \routingMatrix^\top (\routingMatrix \latencySlopeMatrix^{-1} \routingMatrix^\top)^{-1} \demandVec}{\Vert \slantProjMatrix \tollVec + \latencySlopeMatrix^{-1} \routingMatrix^\top (\routingMatrix \latencySlopeMatrix^{-1} \routingMatrix^\top)^{-1} \demandVec \Vert_2}.
\end{align}
\end{theorem}

\begin{remark}
Enforcing additional constraints over the optimization in \eqref{Eqn: Convex Program, Optimal Toll for Distributional Robustness} can compel $\tollVec$ to be within a pre-specified range, at the cost of increasing system-wide latency. 
For instance,
as noted in \cite{ChiuJalotaPavone2026CreditvsDiscount}, the toll designer could implement a second-best tolling mechanism by 
enforcing $\toll_e = 0$
on a pre-specified subset of edges $\edges_{\tollVec} \subset \edges$, to ensure that non-zero tolls are only deployed on edges $e \in \edges \backslash \edges_{\tollVec}$.
\end{remark}

The proof of Thm. \ref{Thm: Optimal Toll for Distributional Robustness} requires 
\cite[Lemma 1, Thm. 4]{Olkin1982TheDistanceBetweenTwoRandomVectorswithGivenDispersionMatrices}, stated below, to characterize latency-maximizing shifts in the disturbance distribution $\disturbDistr$ in the inner optimization of \eqref{Eqn: Problem Statement}.

\begin{proposition}(\cite[Lemma 1, Thm. 4]{Olkin1982TheDistanceBetweenTwoRandomVectorswithGivenDispersionMatrices})
\label{Prop: Linear algebra identity for Gelbrich Hull characterization}
Given $M, \tilde M \in \R^{n \times n}$ with $M, \tilde M \succeq 0$:
\begin{alignat}{2} \nonumber
    \tr\big((\covDisturb^{1/2} \covDisturbNominal \covDisturb^{1/2})^{1/2} \big) = &\max_{C \in \R^{|\edges| \times |\edges|}} \hspace{5mm} &&\tr(C) \\ \nonumber
    &\hspace{0.5cm} \text{s.t.}  &&\begin{bmatrix}
        \covDisturb & C \\
        C^\top & \covDisturbNominal
    \end{bmatrix} \succeq 0.
\end{alignat}
\end{proposition}

\begin{proof}(Proof of Thm. \ref{Thm: Optimal Toll for Distributional Robustness})
From \eqref{Eqn: Problem Statement} in Thm. \ref{Thm: (g) Characterization of Equilibrium System Latency}:
\begin{align} 
\nonumber
    \latencyWorstEq(\disturbVec, \tollVec) &=
    \begin{bmatrix}
        \tollVec^\top & \demandVec^\top
    \end{bmatrix}
    \begin{bmatrix}
        \latencySlopeMatrix & \routingMatrix^\top \\
        \routingMatrix & 0
    \end{bmatrix}^{-1}
    \begin{bmatrix}
        \disturbVec + \tollVec \\ -\demandVec 
    \end{bmatrix}.
\end{align}
First, for ease of notation we separate the terms in \eqref{Eqn: Problem Statement} that depend on $\disturbVec$ from those which do not. Concretely, we define $\latencyWorstSlopeWrtDisturb(\tollVec) \in \R^{|\edges|}$ and $\latencyWorstOffsetWrtDisturb(\tollVec) \in \R$ as follows:
\begin{align} \nonumber
    \latencyWorstSlopeWrtDisturb(\tollVec) &:= \begin{bmatrix}
        I & O
    \end{bmatrix} 
    \begin{bmatrix}
        \latencySlopeMatrix & \routingMatrix^\top \\
        \routingMatrix & O
    \end{bmatrix}^{-1} 
    \begin{bmatrix}
        \tollVec \\ \demandVec
    \end{bmatrix} \\ \label{Eqn: Equilibrium Latency, Slope term w.r.t. alpha}
    &= \slantProjMatrix \tollVec + \latencySlopeMatrix^{-1} \routingMatrix^\top (\routingMatrix \latencySlopeMatrix^{-1} \routingMatrix^\top)^{-1} \demandVec, \\ \nonumber
    \latencyWorstOffsetWrtDisturb(\tollVec) &:= \begin{bmatrix}
        \tollVec^\top & \demandVec^\top
    \end{bmatrix}
    \begin{bmatrix}
        \latencySlopeMatrix & \routingMatrix^\top \\
        \routingMatrix & 0
    \end{bmatrix}^{-1}
    \begin{bmatrix}
        \tollVec \\ -\demandVec 
    \end{bmatrix} \\ 
    \label{Eqn: Equilibrium Latency, Offset term w.r.t. alpha}
    &= \tollVec^\top \slantProjMatrix \tollVec - \demandVec^\top (\routingMatrix \latencySlopeMatrix^{-1} \routingMatrix^\top)^{-1} \demandVec.
\end{align}
Then, by applying \eqref{Eqn: Inverse of B, R top; R, O} to Prop. \ref{Prop: KKT matrix with B and R is invertible, and has psd upper left block}, we find that:
\begin{align} \nonumber
    \latencyWorstEq(\disturbVec, \tollVec) = \latencyWorstSlopeWrtDisturb(\tollVec)^\top \disturbVec + \latencyWorstOffsetWrtDisturb(\tollVec).
\end{align}

The inner maximization in computing the set of distributionally robust optimal tolls, given by \eqref{Eqn: Problem Statement}, can then be computed as follows, with key steps detailed below:
\begin{align} \nonumber
    &\hspace{3mm} \max_{\disturbDistr \in \gelbrichBall\disturbDistrDeviation(\disturbDistrNominal) \cap \M(\disturbRange)} \E_{\disturbVec \sim \disturbDistr}[\latencyWorstEq(\disturbVec, \tollVec)] \\ \nonumber
    &= \max_{\disturbDistr \in \gelbrichBall\disturbDistrDeviation(\disturbDistrNominal) \cap \M(\disturbRange)} \E_{\disturbVec \sim \disturbDistr}\big[ \latencyWorstSlopeWrtDisturb(\tollVec)^\top \disturbVec + \latencyWorstOffsetWrtDisturb(\tollVec) \big] \\ \label{Eqn: Introducing C matrix}
    &= \max_{\substack{\meanDisturb \in \R^{|\edges|} \\ \covDisturb \in \R^{|\edges| \times |\edges|} }}. \hspace{5mm} \latencyWorstSlopeWrtDisturb(\tollVec)^\top \meanDisturb + \latencyWorstOffsetWrtDisturb(\tollVec) \\ \nonumber
    &\hspace{9mm} \text{s.t.} \hspace{3mm} \covDisturb \succeq 0, \\ \nonumber
    & \hspace{1.6cm} \Vert \meanDisturb - \meanDisturbNominal \Vert_2^2 + \tr(\covDisturb + \covDisturbNominal - 2(\covDisturb^{1/2} \covDisturbNominal \covDisturb^{1/2})^{1/2} ) \leq \disturbDistrDeviation^2 \\ 
    \label{Eqn: Applying Gelbrich Hull Identity}
    &= \max_{\substack{\meanDisturb \in \R^{|\edges|} \\ \covDisturb, C \in \R^{|\edges| \times |\edges|} }}. \hspace{5mm} \latencyWorstSlopeWrtDisturb(\tollVec)^\top \meanDisturb + \latencyWorstOffsetWrtDisturb(\tollVec) \\ \nonumber
    &\hspace{9mm} \text{s.t.} \hspace{7.5mm} \covDisturb \succeq 0, \\ \nonumber
    & \hspace{2cm} \Vert \meanDisturb - \meanDisturbNominal \Vert_2^2 + \tr(\covDisturb + \covDisturbNominal - 2C) \leq \disturbDistrDeviation^2, \\ \nonumber
    & \hspace{2cm} \begin{bmatrix}
        \covDisturb & C \\ 
        C^\top & \covDisturbNominal
    \end{bmatrix} \succeq 0, \\
    \label{Eqn: Simplifying Outcome of Applying Gelbrich Hull Identity}
    &= \max_{\meanDisturb \in \R^{|\edges|}}. \hspace{5mm} \latencyWorstSlopeWrtDisturb(\tollVec)^\top \meanDisturb + \latencyWorstOffsetWrtDisturb(\tollVec) \\ \nonumber
    &\hspace{6mm} \text{s.t.} \hspace{1cm} \Vert \meanDisturb - \meanDisturbNominal \Vert_2 \leq \disturbDistrDeviation \\
    \label{Eqn: Applying Robust Linear Programming}
    &= \disturbDistrDeviation \Vert \latencyWorstSlopeWrtDisturb(\tollVec) \Vert_2 + \latencyWorstSlopeWrtDisturb(\tollVec)^\top \meanDisturbNominal + \latencyWorstOffsetWrtDisturb(\tollVec).
\end{align}
To prove $\eqref{Eqn: Introducing C matrix} = \eqref{Eqn: Applying Gelbrich Hull Identity}$, it suffices to show that the feasible set of $(\meanDisturb, \covDisturb)$ for \eqref{Eqn: Introducing C matrix} equals the set of $(\meanDisturb, \covDisturb)$ for which there exists some $C \in \R^{|\edges| \times |\edges|}$ such that $(\meanDisturb, \covDisturb, C)$ is feasible for \eqref{Eqn: Applying Gelbrich Hull Identity}. 
From Prop. \ref{Prop: Linear algebra identity for Gelbrich Hull characterization},
given any $(\meanDisturb, \covDisturb)$ that is feasible for \eqref{Eqn: Introducing C matrix}, the tuple $(\meanDisturb, \covDisturb, C)$ with $C := (\covDisturb^{1/2} \covDisturbNominal \covDisturb^{1/2})^{1/2}$ is feasible for \eqref{Eqn: Applying Gelbrich Hull Identity}. Conversely, given any $(\meanDisturb, \covDisturb, C)$ that is feasible for \eqref{Eqn: Applying Gelbrich Hull Identity}:
\begin{align} \nonumber
    &\Vert \meanDisturb - \meanDisturbNominal \Vert_2^2 + \tr(\covDisturb + \covDisturbNominal - 2(\covDisturb^{1/2} \covDisturbNominal \covDisturb^{1/2})^{1/2}) \\ \nonumber
    \leq &\Vert \meanDisturb - \meanDisturbNominal \Vert_2^2 + \tr(\covDisturb + \covDisturbNominal - 2C) \\ \nonumber
    \leq & 0,
\end{align}
thus confirming that $(\meanDisturb, \covDisturb)$ is feasible for \eqref{Eqn: Introducing C matrix}. Therefore, $\eqref{Eqn: Introducing C matrix} = \eqref{Eqn: Applying Gelbrich Hull Identity}$.
Meanwhile, $\eqref{Eqn: Applying Gelbrich Hull Identity} \geq \eqref{Eqn: Simplifying Outcome of Applying Gelbrich Hull Identity}$ follows by simply taking $\covDisturb = C = \covDisturbNominal$, while $\eqref{Eqn: Applying Gelbrich Hull Identity} \leq \eqref{Eqn: Simplifying Outcome of Applying Gelbrich Hull Identity}$ follows by observing that, given any $\meanDisturb \in \R^{|\edges|}$ for which there exist $\covDisturb, C \in \R^{|\edges| \times |\edges|}$ such that $(\meanDisturb, \covDisturb, C)$ is feasible for \eqref{Eqn: Applying Gelbrich Hull Identity}:
\begin{align} \nonumber
    0 &\leq \tr \left( \begin{bmatrix}
        I & I 
    \end{bmatrix}
    \begin{bmatrix}
        \covDisturb & C \\ 
        C^\top & \covDisturbNominal
    \end{bmatrix}
    \begin{bmatrix}
        I \\ I 
    \end{bmatrix} \right) \\ \nonumber
    &= \tr(\covDisturb + \covDisturbNominal - 2C),
\end{align}
so $\Vert \meanDisturb - \meanDisturbNominal \Vert_2^2 \leq \disturbDistrDeviation^2 - \tr(\covDisturb + \covDisturbNominal - 2C) \leq \disturbDistrDeviation^2$, i.e., $\meanDisturb$ is feasible for \eqref{Eqn: Simplifying Outcome of Applying Gelbrich Hull Identity}.
Finally, $\eqref{Eqn: Simplifying Outcome of Applying Gelbrich Hull Identity} \leq \eqref{Eqn: Applying Robust Linear Programming}$ follows by applying the Cauchy-Schwarz inequality to show that, for any $\meanDisturb \in \R^{|\edges|}$ satisfying $\Vert \meanDisturb - \meanDisturbNominal \Vert_2 \leq \disturbDistrDeviation$:
\begin{align} \nonumber
    \latencyWorstSlopeWrtDisturb(\tollVec)^\top \meanDisturb + \latencyWorstOffsetWrtDisturb(\tollVec) &= \latencyWorstSlopeWrtDisturb(\tollVec)^\top (\meanDisturb - \meanDisturbNominal) + \latencyWorstSlopeWrtDisturb(\tollVec)^\top \meanDisturbNominal + \latencyWorstOffsetWrtDisturb(\tollVec) \\ \nonumber
    &\leq \Vert \meanDisturb - \meanDisturbNominal \Vert_2 \Vert \latencyWorstSlopeWrtDisturb(\tollVec) \Vert_2 + \latencyWorstSlopeWrtDisturb(\tollVec)^\top \meanDisturbNominal + \latencyWorstOffsetWrtDisturb(\tollVec) \\ \nonumber
    &\leq \disturbDistrDeviation \Vert \latencyWorstSlopeWrtDisturb(\tollVec) \Vert_2 + \latencyWorstSlopeWrtDisturb(\tollVec)^\top \meanDisturbNominal + \latencyWorstOffsetWrtDisturb(\tollVec),
\end{align}
while $\eqref{Eqn: Simplifying Outcome of Applying Gelbrich Hull Identity} \geq \eqref{Eqn: Applying Robust Linear Programming}$ follows from noting that the above inequalities all hold with equality when $\meanDisturb := \meanDisturbNominal + \frac{\disturbDistrDeviation}{\Vert \latencyWorstSlopeWrtDisturb(\tollVec) \Vert_2} \latencyWorstSlopeWrtDisturb(\tollVec)$, which is feasible for the constraint set of \eqref{Eqn: Simplifying Outcome of Applying Gelbrich Hull Identity}.
Finally, by substituting \eqref{Eqn: Applying Robust Linear Programming} back into \eqref{Eqn: Problem Statement}, we obtain:
\begin{align} \nonumber
    \optTollSet(\disturbDistrDeviation) &:= \min_{\tollVec \in \tollSet(\disturbDistrDeviation, \disturbRange)} \max_{\disturbDistr \in \gelbrichBall\disturbDistrDeviation(\disturbDistrNominal) \cap \M(\disturbRange)} \E_{\disturbVec \sim \disturbDistr}[\latencyWorstEq(\disturbVec, \tollVec)] \\ \label{Eqn: Optimal Toll set, penultimate expression}
    &= \min_{\tollVec \in \tollSet(\disturbDistrDeviation, \disturbRange)} \disturbDistrDeviation \Vert \latencyWorstSlopeWrtDisturb(\tollVec) \Vert_2 + \latencyWorstSlopeWrtDisturb(\tollVec)^\top \meanDisturbNominal + \latencyWorstOffsetWrtDisturb(\tollVec).
\end{align}
Substituting the definitions of $\latencyWorstSlopeWrtDisturb(\tollVec)$ and $\latencyWorstOffsetWrtDisturb(\tollVec)$, in \eqref{Eqn: Equilibrium Latency, Slope term w.r.t. alpha} and \eqref{Eqn: Equilibrium Latency, Offset term w.r.t. alpha}, transforms \eqref{Eqn: Optimal Toll set, penultimate expression} into \eqref{Eqn: Convex Program, Optimal Toll for Distributional Robustness}, concluding the proof.
\end{proof}


\section{Experiments}
\label{sec: Experiments}

We assess the performance of our distributionally robust tolls, relative to nominal tolls that ignore distributional uncertainty, in reducing system-wide latency on simulations of a 2-link Pigou network with the affine edge latency functions $\ell_1(x) = 1.5 x$ and $\ell_2(x) = 0.1 x$. 
We fix $\delta = 0.2$, and define the nominal disturbance distribution $\disturbDistrNominal$ to be 
uniform
over the ball in $\R^2$ with center at $(20, 30)$ and radius $\delta$.

First, we solve \eqref{Eqn: LP for epsilon max} from Lemma \ref{Lemma: Max Distributional Deviation} to compute $\disturbDistrDeviationMax = 39.8$ to be the maximum deviation of the mean of the disturbance distribution $\disturbDistr$ 
that still ensures robust full network utilization at equilibrium.
We define four levels of disturbance mean shift below $\disturbDistrDeviationMax$, collected into the set $S_\disturbDistrDeviation = \{0.0, 10.0, 20.0, 30.0 \}$. 
Then, we simulate the equilibrium system-wide latency induced when the toll designer deploys a distributionally robust toll with \textit{anticipated} disturbance distribution deviation $\hat\disturbDistrDeviation \in S_\disturbDistrDeviation$, while the disturbance distribution \textit{actually} deviates from its nominal prior by $\disturbDistrDeviation \in S_\disturbDistrDeviation$.
(Note that $\disturbDistrDeviation$ and $\hat \disturbDistrDeviation$ are varied independently over $S_\disturbDistrDeviation$, and thus need not be identical).
Concretely, 
for each pair $(\disturbDistrDeviation, \hat \disturbDistrDeviation) \in S_\disturbDistrDeviation \times S_\disturbDistrDeviation$:
\begin{enumerate}
    \item We optimize the objective in \eqref{Eqn: Convex Program, Optimal Toll for Distributional Robustness}, as presented in Thm. \ref{Thm: Optimal Toll for Distributional Robustness} to compute a distributionally optimal toll with the \textit{anticipated} robustness level $\hat\disturbDistrDeviation$, denoted below by $\tollVec^\star(\hat\disturbDistrDeviation)$.
    \item We use \eqref{Eqn: Worst-case Disturbance Distribution Shift} to characterize the worst-case disturbance distribution $\disturbDistr(\disturbDistrDeviation, \hat\disturbDistrDeviation)$, with \textit{actual} deviation $\disturbDistrDeviation$ in its mean, that corresponds to the deployment of $\tollVec^\star(\hat\disturbDistrDeviation)$, as follows:
    \begin{align}
        \E\big[\disturbDistr(\disturbDistrDeviation, \hat \disturbDistrDeviation )\big] = \meanDisturbNominal + \epsilon \cdot \frac{\latencyWorstSlopeWrtDisturb(\tollVec^\star(\hat\disturbDistrDeviation))}{\Vert \latencyWorstSlopeWrtDisturb(\tollVec^\star(\hat\disturbDistrDeviation)) \Vert_2}.
    \end{align}
    We then draw $N = 10000$ i.i.d. disturbance realizations from the distribution $\disturbDistr(\disturbDistrDeviation, \hat\disturbDistrDeviation)$, denoted below by $\{\disturbVec^\paren{n}(\disturbDistrDeviation, \hat\disturbDistrDeviation): n \in [N] \}$.
    \item Finally, we compute the expected system latency equilibrium corresponding to the deployment of the toll $\tollVec^\star(\hat\disturbDistrDeviation)$ and i.i.d. realizations of the worst-case disturbance distribution $\disturbDistr(\disturbDistrDeviation, \hat\disturbDistrDeviation)$, denoted below by $\bar\latencyWorstEq(\disturbDistrDeviation, \hat\disturbDistrDeviation)$:
    \begin{align}
        \bar\latencyWorstEq(\disturbDistrDeviation, \hat\disturbDistrDeviation) := \frac{1}{N} \sum_{n=1}^N \latencyWorstEq(\tollVec^\star(\hat\disturbDistrDeviation), \disturbVec^\paren{n}(\disturbDistrDeviation, \hat\disturbDistrDeviation)).
    \end{align}
\end{enumerate}
Below, we plot all values $\bar\latencyWorstEq(\disturbDistrDeviation, \hat\disturbDistrDeviation)$ of the expected system latency equilibrium, as $\disturbDistrDeviation$ and $\hat\disturbDistrDeviation$ range independently over $S_\disturbDistrDeviation$:
\begin{align*} 
    \begin{array}{|c||c|c|c|c|}
    \hline
    \disturbDistrDeviation \backslash \hat\disturbDistrDeviation & 0.0 & 10.0 & 20.0 & 30.0  \\ \hline \hline
    0.0 & 3859.42 & 3870.66 & 3900.85 & 3945.00 \\ \hline
    10.0 & 4765.95 & 4754.09 & 4764.16 & 4790.35 \\ \hline
    20.0 & 5672.50 & 5637.52 & 5627.32 & 5635.82 \\ \hline
    30.0 & 6579.02 & 6520.88 & 6490.32 & 6481.12 \\ \hline
\end{array}
\end{align*}

The above results indicate that, 
under any level $\disturbDistrDeviation$ of disturbance distribution shift,
the distributionally robust toll $\tollVec^\star(\hat\disturbDistrDeviation)$ designed with 
$\hat\disturbDistrDeviation = \disturbDistrDeviation$ 
outperforms the distributionally robust toll designed with any $\hat\disturbDistrDeviation \ne \disturbDistrDeviation$ in minimizing system-wide latency.
In particular, for each $\disturbDistrDeviation > 0$, the distributionally robust toll $\tollVec^\star(\disturbDistrDeviation)$ outperforms the nominal toll $\tollVec^\star(0)$ in minimizing system-wide latency when the disturbance distribution undergoes a maximum shift of $\disturbDistrDeviation$.
Thus, our numerical results validate our insight that, in minimizing system-wide latency, distributionally robust tolls outperform tolls designed while ignoring stochastic latency fluctuations.




\section{Conclusion and Future Work}
\label{sec: Conclusion and Future Work}

In this work, we present the distributionally robust toll design problem in single-origin single-destination network congestion games.
We
prove that optimal distributionally robust tolls can be computed via convex programming.
Numerically, we illustrate the efficacy of distributionally robust tolls in reducing system-wide latency, compared against tolls designed without considering stochatic latency fluctuations. 
Future work will develop methods to solve or approximate optimal distributionally robust tolls in more general congestion games (e.g., with polynominal latency maps) and disturbance models.
Further, we will perform numerical experiments over more complex network topologies.

\renewcommand*{\bibfont}{\footnotesize}
\printbibliography

@article{li2023adaptive,
  title={Adaptive constraint satisfaction for markov decision process congestion games: Application to transportation networks},
  author={Li, Sarah HQ and Yu, Yue and Miguel, Nicolas I and Calderone, Dan and Ratliff, Lillian J and A{\c{c}}{\i}kme{\c{s}}e, Beh{\c{c}}et},
  journal={Automatica},
  volume={151},
  pages={110879},
  year={2023},
  publisher={Elsevier}
}

@book{Sandholm2010PopulationGamesAndEvolutionaryDynamics,
  title = {{Population Games And Evolutionary Dynamics}},
  author={William H. Sandholm},
  publisher={{Economic Learning and Social Evolution}},
  year={2010}
}

@article{Ahipasaoglu2019DistributionallyRobustMTE,
author = {Ahipasaoglu, Selin and Arikan, Ugur and Natarajan, Karthik},
year = {2019},
month = {10},
pages = {1546-1562},
title = {{Distributionally Robust Markovian Traffic Equilibrium}},
volume = {53:6},
journal = {Transportation Science},
}

@article{Roughgarden2010AlgorithmicGameTheory,
  title={{Algorithmic Game Theory}},
  author={Roughgarden, Tim},
  journal={Communications of the ACM},
  volume={53},
  number={7},
  pages={78--86},
  year={2010},
  publisher={ACM New York, NY, USA}
}

@article{Yang1998PrincipleOfMarginalCostPricing,
  title={{Principle of Marginal-Cost Pricing: How Does it Work in a General Road Network?}},
  author={Yang, Hai and Huang, Hai-Jun},
  journal={Transportation Research Part A: Policy and Practice},
  volume={32},
  number={1},
  pages={45--54},
  year={1998},
  publisher={Elsevier}
}

@article{Paccagnan2019IncentivizingEfficientUse,
  title={{Incentivizing Efficient Use of Shared infrastructure: Optimal Tolls in Congestion Games}},
  author={Paccagnan, Dario and Chandan, Rahul and Ferguson, Bryce L and Marden, Jason R},
  journal={arXiv preprint arXiv:1911.09806},
  year={2019},
  publisher={arXiv}
}

@article{ChiuJalotaPavone2026CreditvsDiscount,
      title={{Credit-Based vs. Discount-Based Congestion Pricing: A Comparison Study}}, 
      author={Chih-Yuan Chiu and Devansh Jalota and Marco Pavone},
      journal={arXiv preprint arXiv:2602.11077},
      year={2026},
}

@article{Paccagnan2021OptimalTaxesinAtomicCongestionGames,
author = {Paccagnan, Dario and Chandan, Rahul and Ferguson, Bryce L. and Marden, Jason R.},
title = {{Optimal Taxes in Atomic Congestion Games}},
year = {2021},
issue_date = {September 2021},
publisher = {Association for Computing Machinery},
address = {New York, NY, USA},
volume = {9},
number = {3},
issn = {2167-8375},
journal = {ACM Trans. Econ. Comput.},
month = {8},
articleno = {19},
numpages = {33},
}

@article{Rosenthal1973ClassofGames,
author = {Rosenthal, Robert W.},
title = {{A Class of Games Possessing Pure-strategy Nash Equilibria}},
year = {1973},
publisher = {Physica-Verlag GmbH},
address = {DEU},
volume = {2},
number = {1},
journal = {Int. J. Game Theory},
pages = {65–67},
numpages = {3}
}

@article{ChiuFerguson2025RobustnessOfIncentiveMechanisms,
    title = {{Robustness of Incentive Mechanisms Against System Misspecification in Congestion Games}},
    volume = {9},
    journal = {IEEE Control Systems Letters},
    author = {Chiu, Chih-Yuan and Ferguson, Bryce L.},
    year = {2025},
    pages = {276--281},
}

@inproceedings{Kuhn2019WassersteinDistributionallyRobustOptimizationTheoryApplicationsML,
author = {Daniel Kuhn and Peyman Mohajerin Esfahani and Viet Anh Nguyen and Soroosh Shafieezadeh-Abadeh},
title = {{Wasserstein Distributionally Robust Optimization: Theory and Applications in Machine Learning}},
booktitle = {Operations Research \& Management Science in the Age of Analytics},
chapter = {6},
pages = {130-166},
year = {2019}
}

@article{suzdaltsev_distributionally_2022,
    title = {{Distributionally Robust Pricing in Independent Private Value Auctions}},
    volume = {206},
    journal = {Journal of Economic Theory},
    author = {Suzdaltsev, Alex},
    month = {12},
    year = {2022},
    pages = {105555},
}

@article{Olkin1982TheDistanceBetweenTwoRandomVectorswithGivenDispersionMatrices,
title = {{The Distance Between Two Random Vectors with Given Dispersion Matrices}},
journal = {Linear Algebra and its Applications},
author = {I. Olkin},
volume = {48},
pages = {257-263},
year = {1982},
}

@inproceedings{Taskesen2023DistributionallyRobustLQR,
 author = {Taskesen, Bahar and Iancu, Dan and Ko\c{c}yi\u{g}it, \c{C}a\u{g}\i l and Kuhn, Daniel},
 booktitle = {Advances in Neural Information Processing Systems},
 pages = {18613--18632},
 title = {{Distributionally Robust Linear Quadratic Control}},
 volume = {36},
 year = {2023}
}

@ARTICLE{BrownMarden2017RobustnessofMarginalCostTaxesinAffineCongestionGames,
  author={Brown, Philip N. and Marden, Jason R.},
  journal={IEEE Transactions on Automatic Control}, 
  title={{The Robustness of Marginal-Cost Taxes in Affine Congestion Games}}, 
  year={2017},
  volume={62},
  number={8},
  pages={3999-4004},
}

@INPROCEEDINGS{YueFerguson2021IncentiveDesignforCongestionGameswithUnincentivizableUsers,
  author={Yue, Yixiao and Ferguson, Bryce L. and Marden, Jason R.},
  booktitle={2021 60th IEEE Conference on Decision and Control (CDC)}, 
  title={{Incentive Design for Congestion Games with Unincentivizable Users}}, 
  year={2021},
  volume={},
  number={},
  pages={4515-4520},
}

@InProceedings{Gollapudi2023OnlineLearningforTrafficNavigationinCongestedNetworks,
  title = 	 {{Online Learning for Traffic Navigation in Congested Networks}},
  author =       {Gollapudi, Sreenivas and Kollias, Kostas and Maheshwari, Chinmay and Wu, Manxi},
  booktitle = 	 {Proceedings of The 34th International Conference on Algorithmic Learning Theory},
  pages = 	 {642--662},
  year = 	 {2023},
  volume = 	 {201},
}

@article{dokka_robust_2017,
  title={{Robust Toll Pricing: A Novel Approach}},
  author={Dokka, Trivikram and Zemkoho, Alain B. and Gupta, Sonali Sen and Nobibon, Fabrice T.},
  journal={arXiv preprint arXiv:1712.01570},
  year={2017}
}

@article{marecek_distributionally_2017,
  title={{Distributionally Robust Optimisation in Congestion Control}},
  author={Marecek, Jakub and Shorten, Robert and Yu, Jia Yuan},
  journal={arXiv preprint arXiv:1705.09152},
  year={2017}
}

@article{Ratliff2018,
    title = {{A Perspective on Incentive Design: Challenges and Opportunities}},
    year = {2018},
    journal = {Annual Review of Control, Robotics, and Autonomous Systems},
    author = {Ratliff, Lillian J and Dong, Roy and Sekar, Shreyas and Fiez, Tanner},
    number = {1},
    pages = {1--34},
    volume = {2}
}

@article{Niyato2008,
    title = {Competitive {Pricing} for {Spectrum} {Sharing} in {Cognitive} {Radio} {Networks}: {Dynamic} {Game}, {Inefficiency} of {Nash} {Equilibrium}, and {Collusion}},
    volume = {26},
    number = {1},
    journal = {IEEE Journal on Selected Areas in Communications},
    author = {Niyato, Dusit and Hossain, Ekram},
    year = {2008},
    pages = {192--202},
}

@article{Tang2019,
    title = {{Game Theory-based Interactive Demand Side Management Responding to Dynamic Pricing in Price-based Demand Response of Smart Grids}},
    volume = {250},
    journal = {Applied Energy},
    author = {Tang, Rui and Wang, Shengwei and Li, Hangxin},
    year = {2019},
    pages = {118--130},
}

@report{Chiu2026TechnicalNoteTechnicalNoteGraphTheoreticPropertiesofTrafficNetworks,
  title={Technical Note: Graph-Theoretic Properties of Traffic Networks},
  author={Chih-Yuan Chiu},
  year={2026},
  url={https://drive.google.com/file/d/1yaO-KmuoEbXu524MrcNfCSZC-vHd5PuQ/view?usp=sharing}
}

@article{Kim2024EstimateThenPredict,
title = {{Estimate Then Predict: Convex Formulation for Travel Demand Forecasting}},
journal = {(submitted to) Transportation Research Part B: Methodological},
author = {Youngseo Kim and Gioele Zardini and Samitha Samaranayake and Soroosh Shafiee},
year = {2024},
}



\end{document}